\documentclass[journal]{IEEEtran}

\usepackage{cite}
\usepackage{bm}
\usepackage{mathrsfs}
\usepackage{algorithmic}
\usepackage{algorithm}
\usepackage[pdftex]{graphicx}
\usepackage{psfrag}
\usepackage{subfigure}
\usepackage{url}
\usepackage{color}
\usepackage{stfloats}
\usepackage{amsmath}
\usepackage{empheq}
\usepackage[numbers,sort&compress]{natbib}
\usepackage{amsthm}  
\usepackage{amssymb}
\usepackage{array}
\usepackage{balance}
\usepackage{booktabs}
\usepackage{multirow}
\newtheorem{theorem}{\textbf{Theorem}}

\newtheorem{lemma}{\textbf{Lemma}}
\newtheorem{corollary}{\textbf{Corollary}}

\begin{document}

\title{Analyzing Power Beacon Assisted Multi-Source Transmission Using   Markov Chain\\
}

\author{Xuanxuan Tang, Yansha Deng, \emph{Member, IEEE}, Yueming Cai, \emph{Senior Member, IEEE}, Wendong Yang, and Arumugam Nallanathan, \emph{Fellow, IEEE}



\thanks{This work was supported by the National Natural Science Foundation of
China under Grant No. 61771487, and the China
Scholarship Council. This work was done while X. Tang
was a visiting student with the Department of Informatics, King's
College London. The corresponding author is Yansha Deng.}
\thanks{X. Tang, Y. Cai, and W. Yang are with the College of Communications
Engineering, Army Engineering University of PLA, Nanjing 210007, China
(email: tang\_xx@126.com, caiym@vip.sina.com, ywd1110@163.com).}
\thanks{Y. Deng is with the Department of Informatics, King's
College London, London WC2R 2LS, U.K. (email:
yansha.deng@kcl.ac.uk).}

\thanks{A. Nallanathan is with the School of Electronic Engineering and
Computer Science, Queen Mary University of London, London, UK (email:
arumugam.nallanathan@qmul.ac.uk).}
%





}


\maketitle

%
\begin{abstract}
Wireless power transmission (WPT) is envisioned to be a promising technology for prolonging the lifetime of wireless devices in
energy-constrained networks.
This paper presents a general power beacon (PB) assisted multi-source transmission, where a practical source selection scheme with information transmission (IT) mode or non-IT mode is developed to maximize the transmission reliability.
In the IT mode, a zero-forcing (ZF) beamformed signal with no interference to the destination is transmitted  at the multi-antenna PB to supply wireless energy for the sources, and bring non-negative effect to the destination. Among multiple sources, the  energy-sufficient source with the best  channel quality is selected for wireless information transmission (WIT), while the other  sources remain for energy harvesting. In the non-IT mode, the equal power transmission is adopted at PB to focus on energy delivery.
Using Markov chain theory, the energy arrival and departure  of each finite-capacity storage at the source is characterized mathematically,  and the comprehensive analytical expressions of the energy outage probability (EOP), the connection outage probability (COP), and the average transmission delay (ATD) are formulated and derived.
Our results reveal that the EOP, COP, and ATD  can be significantly improved via increasing the number of sources deployed in the proposed network with finite transmit power of PB. We also prove that the multi-source network  will never experience energy outage with infinite transmit power of PB.

\end{abstract}


\section{Introduction}\label{sec1 introduction}


The lifetime of current wireless devices is significantly limited by the energy capacity of their batteries, especially in some energy-constrained scenarios, such as  wireless sensor networks (WSNs) \cite{wsn},  wireless body area networks (WBANs) \cite{Sui2017}, and low-power wide-area networks (LPWANs) \cite{Raza2017}. To cope with this limitation, the energy harvesting (EH)
has emerged as a promising technology to enable the sustainable energy for the devices  without replacing their battery \cite{Dan2015,Qian2016}. Harvesting energy from the ambient environment sources like solar, wind, thermoelectric, electromechanical, etc, has been extensively researched and applied in industry. However, this approach does not suitable for  wireless communication devices, as it could not guarantee controllable  and continuous energy supply due to the randomness and instability of the environment, which may degrade the user experience.

Recently, the wireless power transmission (WPT)  via radio-frequency (RF) radiation has drawn much attention among the wireless communities due to the great advancements of microwave technologies over the past decades \cite{Pan2016,Yang2017}. Compared to  collecting energy  via natural sources,
WPT is capable of delivering a controllable amount of energy as well as  information. In order to realize WPT in practice, the simultaneous wireless information and power transmission (SWIPT) via the same
modulated microwave has been proposed and discussed extensively in existing  literature. Specifically, two classic practical architectures  have been  presented, namely the ``time-switching'' architecture \cite{Zhong2014,ZengandZhang2015} and the ``power-splitting'' architecture \cite{Zhou2013,Wang2017b,Zhou2015}, where the receiving signal can be  splitted  either over time or power domain for independent WPT and wireless information transmission (WIT), respectively. Note that the efficiency of WPT relies on the received signal power while the reliability of WIT hinges on the received signal-to-noise ratio (SNR) \cite{KaibinandXiangyun2015}, as the noise power is relatively low, the distance for WPT is drastically shorter than that for WIT. Under this circumstance,
 the short range limitation of WPT and the same propagation link between WPT and WIT largely limit the  information transmission range.

To face the challenge of  aforementioned SWIPT design, some works have proposed the power beacon (PB) assisted WPT systems \cite{KaibinandXiangyun2015,Zhou2018}. In such systems, WPT and WIT processes are decoupled and  the PB can be deployed much closer to the wireless-powered equipments, which boosts the efficiency of WPT significantly. In 2017, the PB-based product ``Cota Tile'' has been  designed by Ossia Inc. to  charge wireless devices at home, which has received the ``Innovation Awards'' at the 2017 Consumer Electronics Show (CES) \cite{CotaCES2017}. In \cite{Jiang2016}, a novel PB assisted wiretap channel was studied to exploit secure communication between the energy constrained source and  a legitimate user under eavesdropping. In  \cite{Ma2015}, the authors studied a PB assisted wireless-powered
system, where each user first harvested energy from RF signals broadcast by its associated AP and/or the PB in the downlink and then used the harvested energy for information transmission in the uplink. The faction of the time duration for downlink WPT was then optimized for each user.
In \cite{Chen2017}, the users clustering around the PB for WPT, and deliver information to the APs.
In \cite{Shi2018}, the  device-to-device (D2D) communication sharing the resources of downlink cellular network was powered by the energy from PBs.


 The aforementioned works
 have assumed no energy storage across different time slots, and
the wireless-powered devices  consume all the harvested energy
in the current time slot to perform its own information transmission (IT). This type of operating mode, named as ``harvest-use'' \cite{Salem}, may  not practical due to the following two-fold reasons: 1) the wireless devices are usually equipped with battery, which can storage energy over different time slots; and  2) the ``harvest-use'' approach  results in the random fluctuation of
the instant transmit power of a wireless-powered device,
which may not only affect the performance of the  device itself, but also  create chaos  to the whole system.

Recent research have shifted to the so-called ``harvest-store-use''  operating mode \cite{Pielli2017}, where the devices are capable of storing the harvested energy in a rechargeable
battery.
In \cite{BiandChen2016}, an accumulate-and-jam protocol was presented to enhance the physical layer security in wireless transmission. The full-duplex (FD) relaying scheme was studied in the WPT networks in \cite{Liu2016a,Liu2017}. Unlike the single energy storage scenario in  \cite{BiandChen2016,Liu2016a,Liu2017},   the multiple energy storages were considered in \cite{AhmedPIMRC2015} and \cite{Yao2017} with energy harvested from natural sources and wireless signals, respectively.
In \cite{Yao2017}, a wireless-powered uplink and downlink network was studied,  where the WPT occurs in the downlink  performing by time division AP, and the IT occurs in both uplink and downlink, where the AP transmits the downlink information, and  the users use the harvested energy storing for uplink information. Both the energy at the AP and the  users are modelled using Markov chain,  and the time-frequency resource allocation and user scheduling problem was studied to minimize overall energy consumption.

In this paper, we study the energy storage and data transmission  of PB assisted wireless-powered multi-source networks, where the network can operates in  the non-IT mode or the IT mode. In the non-IT without any energy sufficient sources,  the whole network experiences  energy outage event with no IT in the network, thus the PB uses the equal power allocation among all antennas for directing its energy. In the IT mode, the source with the best  channel quality among all the energy-sufficient sources is selected for IT, to avoid the interference from the PB to the destination, a zero-forcing (ZF) beamformed signal is designed during the WPT, with no inteference to the data transmission between the source and the destination.


\begin{itemize}
  \item
  We formulate a  Markov-based analytical framework for the energy storage and energy usage of  the proposed multi-source wireless-powered networks to characterize its dynamic behaviors of the energy arrival and departure.
  To facilitate the network performance analysis, we also derive the state transition probabilities of the proposed network, and the stationary probabilities of all the states.
  \item  We propose a operating mode selection procedure for the proposed network to select the  non-IT mode and  the IT mode based on the energy states of all the sources. A flexible beamforming transmission scheme is proposed at the PB, which can adapt to the network operating mode.  In the non-IT mode, the beamformer is designed with equal power among antennas. In the IT mode, the beamformer is designed to bring no interference to the data transmission between the selected source and the destination.
  \item  Based on our derived stationary probabilities of all the states,  we derive the energe outage probability and the connection outage probability of proposed network in the non-IT mode and the IT mode, respectively.
  To quantify the delay performance, we also define and derive anlaytical expression for the average transmission delay of proposed networks. Our derived analytical results are all validated via simulation, which show the correctness of our derivations, and demonstrate design insights.

\end{itemize}

The remainder of the work is organized as follows: Section \ref{sec2
systemmodel} describes the system model and
presents the details of the operating mode selection as well as the source selection procedure. In \ref{sec4}, the energy state transitions among all the states are carried out.
In \ref{sec5}, the energy outage probability, the connection outage probability, and the average transmission delay  of the  network are respectively investigated. Simulation results are given in Section
\ref{sec6 NUMERICAL RESULTS}, and Section \ref{sec7 Conclusions}
summarizes the contributions of this paper.

\emph{\textbf{Notation:}} Throughout this paper, the boldface uppercase letters are used to denote matrices or vectors. ${{\left( \cdot  \right)}^{T}}$, ${{\left( \cdot  \right)}^{H}}$, and ${{\left( \cdot  \right)}^{\dagger }}$ are denoted as the transpose operation, the conjugate transpose operation, and the orthogonal operation, respectively. ${{F}_{\gamma }}\left( \cdot  \right)$ and ${{f}_{\gamma }}\left( \cdot \right)$ represent the cumulative distribution function (CDF) and the probability density function (PDF) of random variable $\gamma$, respectively. $\mathbb{E}\left[ \cdot  \right]$ denotes the expectation operation.



\section{System Model}\label{sec2 systemmodel}
We consider a multi-source wireless-powered transmission network as shown in Fig. \ref{figsys}, which consists of single power beacon node $B$, $K$ number of  wireless-powered source nodes $\left\{ {{S_k}} \right\}_{k = 1}^K$, and a destination node $D$. It is assumed that $B$ is equipped with ${{N}_{B}}$ antennas,  and all the other nodes are equipped with a single antenna, where all nodes are working in half-duplex (HD) mode. Each source (IoT/ mobile device) is equipped with an energy storage with a finite energy capacity of ${{\varepsilon }_{T}}$.
We assume  that all the channels experience quasi-static Rayleigh fading and the channel coefficients keep constant during a block time ${{T}_{0}}$ but change independently from one packet time to another\footnote{This assumption has been extensively adopted in the WPT researches \cite{BiandChen2016,Krikidis2014,Zhang2016}.}. A standard path-loss model \cite{Hosseini2014,Zhong2014} is adopted, namely the average channel power gain $\bar\gamma_{ab}=\mathbb{E}\left[ {{{\left| {{h_{ab}}} \right|}^2}} \right] = d_{ab}^{ - \alpha }$, where $\alpha$ is the path-loss factor, $h_{ab}$ and $d_{ab}$ denote the channel coefficient and the distance between $a$ and $b$, respectively.

\begin{figure}
\begin{center}
  \includegraphics[width=3.5 in,angle=0]{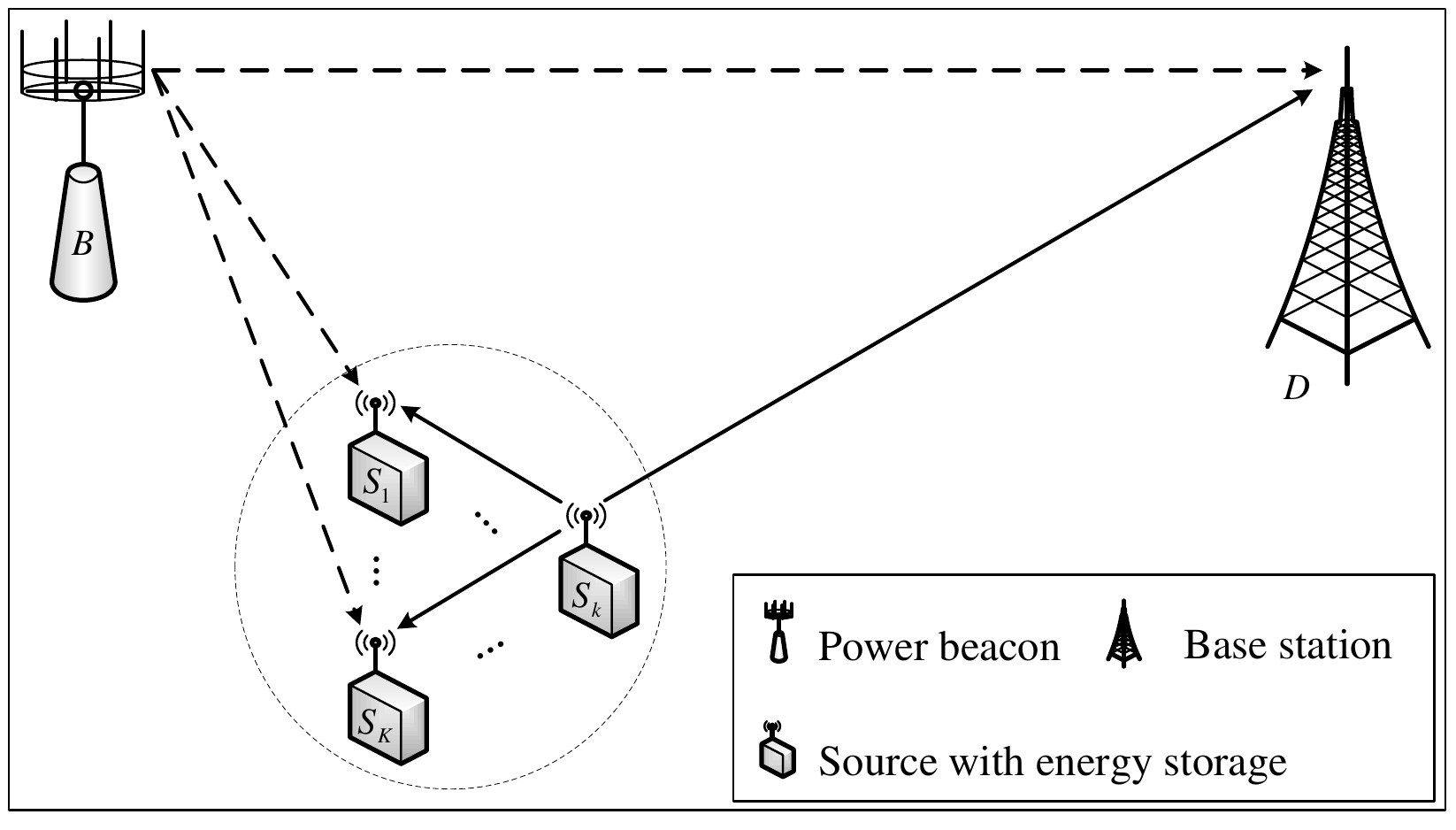}\\
  \caption{System model}\label{figsys}
\end{center}
\end{figure}

\subsection{Energy Discretization and State Modeling}
To quantify the energy storage at the sources, we define a discrete-level model \cite{BiandChen2016,Liu2017}, namely, each storage is discretized into $1+L$ levels, where $L$ is the discretizing level of the network, and  the $l$-th energy level is defined as
\begin{align}\label{eq1}
{\varepsilon _l} = \begin{array}{*{20}{c}}
{l \cdot \varepsilon _\Delta,}&{l \in \left\{ {0,1, \cdots ,L} \right\}}
\end{array},
\end{align}
where ${\varepsilon _\Delta } = {\textstyle{{{\varepsilon _T}} \over L}}$ is the single unit of energy. For instance, if the  new energy arrival from the harvested energy  at the $k$th source node  is $\varepsilon _{k}$, the amount of energy that can be saved in the energy storage after discretization can be expressed as \cite{BiandChen2016,Liu2017}
\begin{align}\label{eqrule}
\tilde \varepsilon _{k} = {\varepsilon _{{l^*}}},{\rm{with}}\;{l^{*}} = \mathop {\arg \max }\limits_{l \in \left\{ {0,1, \cdots ,L} \right\}} \left\{ {{\varepsilon _l}:{\varepsilon _l} \le \varepsilon _{k}} \right\}.
\end{align}

 Recall  that there are $K$ storages and each storage has $(1+L)$ levels, thus we have $N$  states in total with $N={{\left( 1+L \right)}^{K}}$.
The energy level indexes in all the storages  form an energy state set $\Theta = \{ { {{{\vec s}_1}, \cdots ,{{\vec s}_n}, \cdots ,{{\vec s}_N}}  }\}$, where  the $n$th state is given by
\begin{align}\label{eq2}
{\vec s_n} = \left[ {l_1^n, \cdots ,l_k^n, \cdots ,l_K^n} \right],
\end{align}
with $n\in \left\{ 1,\cdots ,N \right\}$ and $l_{k}^{n} { \in \left\{ {0,1, \cdots ,L} \right\}}$ representing the energy level index of $k$th storage at state ${{\vec s}_{n}}$.

\subsection{Network Operating Modes }\label{sec3}
At any given time, the network is under one specific energy state, and different operating modes are adopted in different states.
When the source does not have enough energy to support the IT operation, we define this source as energy outage. When all the sources experience the energy outage at the same time, we define the multi-source network as energy outage.
As a result, we assume two operating modes: 1) the network operates in IT mode when there is as least one source can perform IT operation; and 2) the network operates in non-IT mode when all the sources are in energy outage.
Next, we describe the operating mode selection procedure, and  the transmission formulation for each operating mode.

\subsubsection{Operating Mode Selection}
We select the operating mode based on the  distributed selection method \cite{Bletsas2006,Tang2017}.
At the start of each time slot, a pilot signal is broadcasted by  $D$.  Using this pilot
      signal, all the sources that are not in energy outage as well as the PB can individually estimate the channel power gains between themselves and $D$. For each source that is not in energy outage, its timer with a parameter inversely
      proportional to its own channel power gain  is switched on,
 namely, the timer of source $S_k$ has the
      parameter of ${{{C_0}} \mathord{\left/
 {\vphantom {{{C_0}} {{{\left| {{h_{{S_k}D}}} \right|}^2}}}} \right.
 \kern-\nulldelimiterspace} {{{\left| {{h_{{S_k}D}}} \right|}^2}}}$, where  $h_{{S_k}D}$ denotes the channel coefficient between $S_k$ and $D$, and
       $C_0$ is a constant and is properly set to ensure that the shortest duration among all the timers always finishes
       within the given duration \cite{Bletsas2006,Tang2017}. Once the shortest timer expires, the corresponding source sends a short flag signal to declare its
       existence,  and all the other sources who are  waiting for
       their timer expiring will back off when they
      hear this flag signal from another source and start to harvest energy. At the same time,  $D$ will get ready for receiving useful information upon hearing this flag signal.

For the source that is in energy outage, it will neither estimate its channel nor set a timer. Hence, if  the whole network undergoes energy outage, no flag signal would be produced during this flag signal duration. As a result, the operating mode of the network can be easily determined and known by all the nodes within the network. For the notation convenience, the set of indexes of source nodes that are in IT mode at state ${{\vec s}_{n}}$ are defined as
\begin{equation}\label{}
\vartheta _{n}^{TH}=\left\{  k:l_{k}^{n}\ge l_{S}^{th} \right\},
\end{equation}
 where $l_{S}^{th}$ denotes the transmit energy level threshold, which is expressed as
      \begin{align}\label{lsth}
       l_S^{th} = \mathop {\arg \min }\limits_{l \in \left\{ {1, \cdots ,L} \right\}} {\kern 1pt} \left\{ {{\varepsilon _l}:{\varepsilon _l} \ge \varepsilon _S^{th}} \right\},
        \end{align}
        where $\varepsilon _{S}^{th}$ denotes the transmit energy threshold of sources.

\subsubsection{Non-IT Operating Mode}
When the network remains at the Non-IT operating mode, we have $\vartheta _{n}^{TH}={{\Phi }_{0}}$, where ${{\Phi }_{0}}$ is the empty set. As described above, no information could be transmitted and all of the sources will harvest energy from the wireless signal transmitted by $B$. Specifically, the harvested energy at the $k$-th source is expressed as
\begin{align}\label{eq4}
\varepsilon _{k}^{n}=\eta {{T}_{0}}{{P}_{B}}{{\left| \mathbf{h}_{B{{S}_{k}}}^{T}{{\mathbf{w}}_{1}} \right|}^{2}},
\end{align}
where ${{\bf{h}}_{B{S_k}}} = {\left[ {{h_{{B_1}{S_k}}}, \cdots ,{h_{{B_b}{S_k}}}, \cdots ,{h_{{B_{N_B}}{S_k}}}} \right]^T}$ represents the channel coefficient vector between $B$ and ${{S}_{k}}$, $k\in \left\{ 1,\cdots ,K \right\}$, $b\in \left\{ 1,\cdots ,N_B \right\}$. ${{\mathbf{w}}_{1}}\in {{\mathbb{C}}^{{{N}_{B}}\times 1}}$ is the normalized weight vector applied at $B$ with its $b$th element satisfying ${{w}_{1,b}}={1}/{\sqrt{{{N}_{B}}}}\;$. The amount of harvested energy that can be saved in the $k$th energy storage after discretization, $ \tilde \varepsilon _k^{n}$, can be obtained according to \eqref{eqrule} by making an appropriate replacement, namely $\varepsilon _k \to \varepsilon _k^{n}$, $\tilde \varepsilon _k \to \tilde \varepsilon _k^{n}$.

\subsubsection{ IT Operating Mode}
When the network remains at the IT operating mode, we have $\vartheta _{n}^{TH} \ne {{\Phi }_{0}}$. As such,  a source that has the largest channel power gain  is selected for IT operation among all the satisfied sources. Mathematically, the index of the selected source can be given by
\begin{align}\label{eq6}
{{i}^{*}}=\arg \underset{k\in \vartheta _{n}^{TH}}{\mathop{\max }}\,\left\{ {{\left| {{h}_{{{S}_{k}}D}} \right|}^{2}} \right\}.
\end{align}

Enjoying the energy harvested from   $B$, the source to destination transmission may also suffer from interference brought by  the wireless signals delivered by $B$.
Note that $B$ has also estimated the channel between itself and $D$ with the pilot signal. To exploit the advantages of multiple antennas, the ZF beamforming scheme can be used at $b$ to fully avoid the interference from $B$ to $D$. To be specific, a normalized weight vector $\mathbf{w}_{2}\in {{\mathbb{C}}^{{{N}_{B}}\times 1}}$ satisfying $\mathbf{w}_{2}=\mathbf{h}_{BD}^{\dagger }$ is applied at $B$ so as to keep $\mathbf{h}_{BD}^{T}\mathbf{w}_{2}\text{=}0$, where ${{\bf{h}}_{B{D}}} = {\left[ {{h_{{B_1}{D}}}, \cdots ,{h_{{B_b}{D}}}, \cdots ,{h_{{B_{N_B}}{D}}}} \right]^T}$ represents the channel coefficient vector between $B$ and $D$, and ${{\left( \cdot  \right)}^{\dagger }}$ denotes the orthogonal operation. Hence, the received signal-to-noise ratio (SNR) at $D$ is given by
\begin{align}\label{eq7}
\gamma _{D}^{n,{i^ * }} = \frac{{{P_S}}}{{N_0}}{\left| {{h_{{S_{{i^ * }}}D}}} \right|^2},
\end{align}
where $N_0$ is the power density of the additive white Gaussian noise (AWGN),
and $P_S$ represents the transmit power of sources with
\begin{align}
{P_S} = {l_S}\frac{{{ \varepsilon _\Delta}}}{T_0},
\end{align}
where $l_S$ is  the actual transmit energy level satisfying $l_S^{th} \le {l_S} \le L$.



At the same time, all the other sources except ${{S}_{{{i}^{*}}}}$  harvest energy from wireless signals, and the harvested energy at the $k$th source on condition that $S_{i^*}$ is selected for IT could be expressed as
\begin{align}\label{}
\varepsilon _{k}^{n,{{i}^{*}}}=\eta {{T}_{0}}\left( {{P}_{S}}{{\left| {{h}_{{{S}_{{{i}^{*}}}}{{S}_{k}}}} \right|}^{2}}+{{P}_{B}}{{\left| \mathbf{h}_{B{{S}_{k}}}^{T}\mathbf{w}_{2} \right|}^{2}} \right),
\end{align}
where $k\in \left\{ 1,\cdots ,K \right\}\backslash \left\{ {{i}^{*}} \right\}$, and the amount of harvested energy that can be saved in the $k$th energy storage after discretization, $\tilde \varepsilon _k^{n,{i^*}}$, can be derived according to \eqref{eqrule} by making an appropriate replacement, namely $\varepsilon _k \to \varepsilon _k^{n,{i^*}}$, $\tilde \varepsilon _k \to \tilde \varepsilon _k^{n,{i^*}}$.

\section{Energy State Transitions}\label{sec4}
In this section, we present a thorough study on the transitions of the energy states. Let us denote
	${{\vec s}_{n}}$
 and
	${{\vec s}_{{{n'}’}}}$
as the states at the current and the next time slots, respectively, $n,{n'}\in \left\{ 1,\cdots ,N \right\}$. We then denote the transition probability to transfer from 	${{\vec s}_{n}}$ to	${{\vec s}_{{{n'}}}}$ within one step  as ${{p}_{{{\vec s}_{n}}\to {{\vec s}_{{{n}'}}}}}$.
For the notation convenience, the
non-IT set of states $\Theta_1$ and the IT set of states $\Theta_2$ are defined as
\begin{align}\label{eqfai1}
{{\Theta }_{1}}=\left\{ {{\vec s}_{n}}:\forall n,s.t.\vartheta _{n}^{TH}={{\Phi }_{0}} \right\},
\end{align}
and
\begin{align}\label{eqfai2}
{{\Theta }_{2}}=\left\{ {{\vec s}_{n}}:\forall n,s.t.\vartheta _{n}^{TH}\ne {{\Phi }_{0}} \right\},
\end{align}
respectively.
Note that ${{\Theta }_{1}}$ represents the set of states that all the sources have to conduct EH operation.
In other words,
we have
${\Theta _1} = \left\{ {\left[ {0, \cdots ,0,0} \right], \cdots ,\left[ {\left( {{l_S} - 1} \right), \cdots ,\left( {{l_S} - 1} \right),\left( {{l_S} - 1} \right)} \right]} \right\}$. Besides, ${{\Theta }_{2}}$ is the set of states that at least one source can perform  IT operation. It is obvious that ${{\Theta }_{1}}$ is the complement set of ${{\Theta }_{2}}$,
so the numbers of states in $\Theta_1$ and $\Theta_2$ are $N_{1}=({l_S})^K$ and $N_{2}=N-({l_S})^K$, respectively. The energy level increment between two states is defined as
\begin{align}\label{}
\Delta {{\vec s}_{n}}={{\vec s}_{n'}}-{{\vec s}_{n}}=\left[ \Delta l_{1}^{n},\cdots ,\Delta l_{k}^{n},\cdots ,\Delta l_{K}^{n} \right],
\end{align}
where
	$\Delta l_{k}^{n}=l_{k}^{{{n}'}}-l_{k}^{n}$
, and
	$l_{k}^{{{n}'}},l_{k}^{n}\in \left\{ 0,1,\cdots ,L \right\}$
.



Note that when the network operates in the non-IT mode, the energy level in any of the sources will not decline. Hence, it is not possible to transfer from ${{\vec s}_{n}}$ to ${{\vec s}_{{{n'}}}}$ within one step if ${{\vec s}_{{{n'}}}}\notin \Theta _{1}^{n,{n}'}$, where $\Theta _{1}^{n,{n}'}$ is a subset of $\Theta$ which satisfies $\Theta _{1}^{n,{n}'}=\left\{ {{\vec s}_{{{n'}}}}:\forall k,s.t.\Delta l_{k}^{n}\ge 0 \right\}$. It is noted that the construction of $\Theta _{1}^{n,{n}'}$ relies on a specific ${{\vec s}_{n}}$. Similarly, when the network operates in the IT mode, the energy level in any of the sources except for the selected one $S_{i^*}$ will not decline, and the energy level of $S_{i^*}$ will decline $l_{S}^{th}$ due to the IT operation.
Hence, it is not possible to transfer from ${{\vec s}_{n}}$ to ${{\vec s}_{{{n'}}}}$ within one step if ${{\vec s}_{{{n}'}}}\notin \Theta _{2}^{n,{n}',i^*}$, where $\Theta _{2}^{n,{n}',i^*}$ is a subset of $\Theta$ which satisfies $\Theta _{2}^{n,{n}',i^*}=\left\{ {{\vec s}_{{{n}'}}}:\Delta l_{i^*}^{n}=-l_{S}^{th};\forall k\ne i^*,\Delta l_{k}^{n}\ge 0 \right\}$.

\subsection{State Transition When Operating in Non-IT Mode}
If the network works in the non-IT mode, namely ${{\vec s}_{n}}\in {{\Theta }_{1}}$. Then the probability of the $k$th source ${{S}_{k}}$ to transfer from $l_{k}^{n}$ to  $l_{k}^{{{n}'}}$ within one step can be expressed as
\begin{align}\label{}
p_{k}^{n\to {n}'}=\Pr \left( {{\vec s}_{{{n}'}}}\left( k \right)=l_{k}^{{{n}'}}|{{\vec s}_{n}}\left( k \right)=l_{k}^{n} \right).
\end{align}

\begin{figure}
\begin{center}
  \includegraphics[width=3.0in,angle=0]{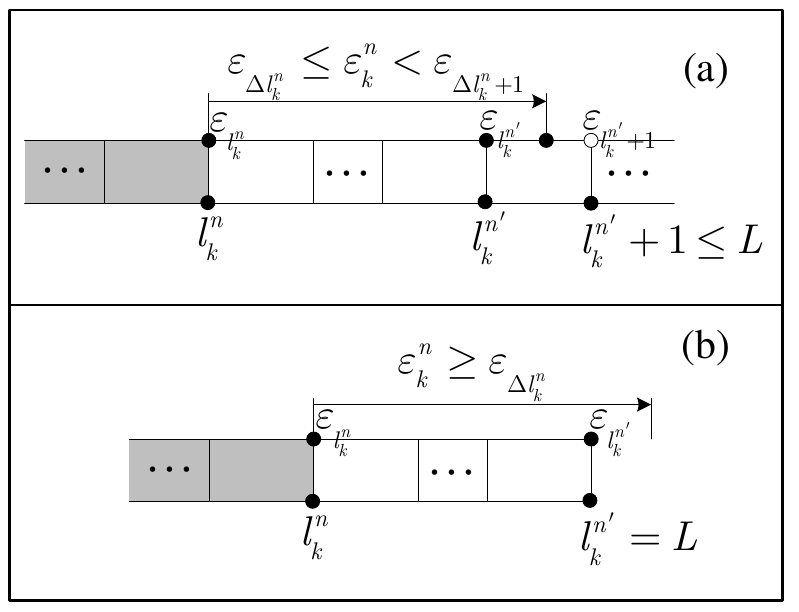}\\
  \caption{State transition when (a) $l_k^{n'} \ne L$ and (b) $l_k^{n'} = L$.}\label{figtransfer}
\end{center}
\end{figure}

As mentioned before, $p_{k}^{n\to {n}'}\text{=}0$ is always true when $\Delta l_{k}^{n}<0$ in non-IT mode. For the case with $\Delta l_{k}^{n}\ge 0$, we will show that $p_{k}^{n\to {n}'}$ differs when $l_{k}^{{{n}'}}=L$ and $l_{k}^{{{n}'}}\ne L$, which imply the state after transition for ${{S}_{k}}$ is full and not full, respectively. On one hand, if  $l_{k}^{{{n}'}}\ne L$, as shown in Fig. \ref{figtransfer} (a), $l_{k}^{n}$ can transfer to $l_{k}^{{{n}'}}$ within one step only when the harvested energy, $\varepsilon _{k}^{n}$, satisfying ${{\varepsilon }_{\Delta l_{k}^{n}}}\le \varepsilon _{k}^{n}<{{\varepsilon }_{\Delta l_{k}^{n}+1}}$. 
On the other hand, if $l_{k}^{{{n}'}}=L$, as shown in Fig. \ref{figtransfer} (b),  it will transfer from $l_{k}^{n}$ to  $l_{k}^{{{n}'}}$ within one step only when the harvested energy $\varepsilon _{k}^{n}$ satisfying $\varepsilon _{k}^{n}\ge {{\varepsilon }_{\Delta l_{k}^{n}}}$. Hence, we formulate the transition probability of the $k$th source as
\begin{align}\label{eq44}
p_{k}^{n\to {n}'}=\left\{ \begin{matrix}
   {{F}_{\varepsilon _{k}^{n}}}\left( {{\varepsilon }_{\Delta l_{k}^{n}+1}} \right)-{{F}_{\varepsilon _{k}^{n}}}\left( {{\varepsilon }_{\Delta l_{k}^{n}}} \right), & l_{k}^{{{n}'}}\ne L,  \\
   1-{{F}_{\varepsilon _{k}^{n}}}\left( {{\varepsilon }_{\Delta l_{k}^{n}}} \right), & l_{k}^{{{n}'}}=L.  \\
\end{matrix} \right.
\end{align}

When ${{\vec s}_{n}}\in {{\Theta }_{1}}$, as each source harvests energy independently, the transition probability of all the sources $p_{{{\vec s}_{n}}\to {{\vec s}_{{{n}'}}}}^{1}$ can be expressed as
\begin{align}\label{eq24}
p_{{{\vec s}_{n}}\to {{\vec s}_{{{n}'}}}}^{1}=\prod\limits_{k=1}^{K}{p_{k}^{n\to {n}'}}.
\end{align}
To derive $p_{{{\vec s}_{n}}\to {{\vec s}_{{{n}'}}}}^{1}$, we present the following Lemma.
\begin{lemma}\label{lemma1}
The CDF of   energy harvested at the $k$th source with the $n$th state $\varepsilon _{k}^{n}$ is derived as
\begin{align}\label{eqlemma1}
{{F}_{\varepsilon _{k}^{n}}}\left( x \right)=1-\exp \left( -\frac{x}{\eta {{T}_{0}}{{P}_{B}}{{{\bar{\gamma }}}_{B{{S}_{k}}}}} \right).
\end{align}
\end{lemma}

\begin{IEEEproof}
We first present the CDF as
\begin{align}\label{eq26}
{{F}_{\varepsilon _{k}^{n}}}\left( x \right)= & {\rm Pr}(\eta {{T}_{0}}{{P}_{B}}{{\left| \mathbf{h}_{B{{S}_{k}}}^{T}{{\mathbf{w}}_{1}} \right|}^{2}} \le x) \notag\\
= &{\rm Pr}(\eta {{T}_{0}}{{P}_{B}}{{\left| \sum\nolimits_{b=1}^{{{N}_{B}}}{{{h}_{{{B}_{b}}{{S}_{k}}}}{{w}_{1,b}}} \right|}^{2}} \le x) \notag\\
= &{\rm Pr}(\eta {{T}_{0}}X_1 \le x)
= \Pr \left( {{X}_{1}} \le \frac{x}{\eta {{T}_{0}}} \right)
\end{align}
where ${{X}_{1}}={{P}_{B}}{{\left| \sum\nolimits_{b=1}^{{{N}_{B}}}{{{h}_{{{B}_{b}}{{S}_{k}}}}{{w}_{1,b}}} \right|}^{2}}$. With Rayleigh fading ${{h}_{{{B}_{b}}{{S}_{k}}}}\sim \mathcal{C}\mathcal{N}\left( 0,{{{\bar{\gamma }}}_{B{{S}_{k}}}} \right)$ and  ${{w}_{1,b}}\text{=}{1}/{\sqrt{{{N}_{B}}}}\;$ with $b\in \left\{ 1,\cdots ,{{N}_{B}} \right\}$, we have $\sum\nolimits_{b=1}^{{{N}_{B}}}{{{h}_{{{B}_{b}}{{S}_{k}}}}{{w}_{1,b}}}\text{=}\frac{1}{\sqrt{{{N}_{B}}}}\sum\nolimits_{b=1}^{{{N}_{B}}}{{{h}_{{{B}_{b}}{{S}_{k}}}}}$. It is noted that the sum of finite Gaussian random variables is still a Gaussian random variable \cite{Ding2017}, hence $\sum\nolimits_{b=1}^{{{N}_{B}}}{{{h}_{{{B}_{b}}{{S}_{k}}}}}\sim \mathcal{C}\mathcal{N}\left( 0,{{N}_{B}}{{{\bar{\gamma }}}_{B{{S}_{k}}}} \right)$ , which results in $\frac{1}{\sqrt{{{N}_{B}}}}\sum\nolimits_{b=1}^{{{N}_{B}}}{{{h}_{{{B}_{b}}{{S}_{k}}}}}\sim \mathcal{C}\mathcal{N}\left( 0,{{{\bar{\gamma }}}_{B{{S}_{k}}}} \right)$. Therefore, ${{X}_{1}}$ is an exponentially distributed random variable with the mean of ${{P}_{B}}{{\bar{\gamma }}_{B{{S}_{k}}}}$ with
\begin{align}\label{eq25}
{{f}_{{{X}_{1}}}}\left( y \right)=\frac{1}{{{P}_{B}}{{{\bar{\gamma }}}_{B{{S}_{k}}}}}\exp \left( -\frac{y}{{{P}_{B}}{{{\bar{\gamma }}}_{B{{S}_{k}}}}} \right).
\end{align}
Substituting  \eqref{eq25} into \eqref{eq26}, we prove \eqref{eqlemma1}.
\end{IEEEproof}

By substituting \eqref{eq44} and \eqref{eqlemma1}  into \eqref{eq24}, the transition probability of all sources from ${{\vec s}_{n}}$ to  ${{\vec s}_{{{n'}}}}$ in non-IT mode for ${{\vec s}_{{{n}'}}}\in \Theta _{1}^{n,{n}'}$ can be derived as
\begin{align}\label{eq38}
p_{{\vec s_n} \to {\vec s_{n'}}}^1 = {e^{ - \frac{{{\varepsilon _{\Delta l_k^n}}}}{{\eta {T_0}{P_B}{{\bar \gamma }_{B{S_k}}}}}}}\prod\limits_{k \notin \vartheta _{n'}^L} {_{k = 1}^K} \left( {1 - {e^{ - \frac{{{\varepsilon _1}}}{{\eta {T_0}{P_B}{{\bar \gamma }_{B{S_k}}}}}}}} \right).
\end{align}

\subsection{State Transition When Operating in IT Mode}
If the network works in IT-mode, namely $\vec s_n \in \Theta_2$.
We further assume the condition that $i^*=i$, namely the transfer from $\vec s_n$ to $\vec s_{n'}$ results from the IT selection of $S_i$. Correspondingly, we have $\Theta _{2}^{n,{n}',i^*}=\Theta _{2}^{n,{n}',i}=\left\{ {{\vec s}_{{{n}'}}}:\Delta l_{i}^{n}=-l_{S}^{th};\forall k\ne i,\Delta l_{k}^{n}\ge 0 \right\}$ on this condition.
Then the transition probability of the $k$th source to transfer from ${{\vec s}_{n}}$ to  ${{\vec s}_{{{n'}}}}$  within one step is derived as
\begin{align}\label{}
p_{k}^{n\to {n}',i}=\left\{ \begin{matrix}
   {{F}_{\varepsilon _{k}^{n,i^{}}}}\left( {{\varepsilon }_{\Delta l_{k}^{n}+1}} \right)-{{F}_{\varepsilon _{k}^{n}}}\left( {{\varepsilon }_{\Delta l_{k}^{n}}} \right), & l_{k}^{{{n}'}}\ne L,  \\
   1-{{F}_{\varepsilon _{k}^{n,i^{}}}}\left( {{\varepsilon }_{\Delta l_{k}^{n}}} \right), & l_{k}^{{{n}'}}=L.  \\
\end{matrix} \right.
\end{align}
As such, the transition probability of all sources from ${{\vec s}_{n}}$ to  ${{\vec s}_{{{n'}}}}$ in the IT mode can be written as
\begin{align}\label{}
p_{{{\vec s}_{n}}\to {{\vec s}_{{{n}'}}}}^{2}=p_{i^{}}^{n}\prod\limits_{k \ne i }{_{k=1}^{K}p_{k}^{n\to {n}',i^{}}}.
\end{align}

To derive the source selection probability $p_{i}^{n}$, we present the following Lemma.
\begin{lemma}\label{lemma3}
For the IT mode ${{\vec s}_{n}}\in {{\Theta }_{2}}$,  the probability that the source ${{S}_{i}}$ satisfying $l_{i}^{n}\ge l_{S}^{th}$ to be
selected for information transmission $p_{i}^{n}=\Pr \left( {{S}_{{{i}^{*}}}}={{S}_{i}} \right)$ is derived as
\begin{align}\label{eq33}
p_{i}^{n}=\sum\limits_{{{\mathbf{n}}_{1}}\in \tau _{1,K}^{i}}{\frac{{{\left( -1 \right)}^{\sum\nolimits_{k=1}^{K}{{{n}_{1,k}}}}}}{{{{\bar{\gamma }}}_{{{S}_{i}}D}}\sum\nolimits_{k=1}^{K}{\frac{{{n}_{1,k}}}{{{{\bar{\gamma }}}_{{{S}_{k}}D}}}}+1}},
\end{align}
where $\tau _{1,K}^{i}$ is the set of $K-$length vectors with all its elements as binary numbers, ${{\mathbf{n}}_{1}}$ is a qualified vector in $\tau _{1,K}^{i}$ with its $k$th element satisfying ${{n}_{1,k}}\in \left\{ 0,1 \right\}$ for $k\in \vartheta _{n}^{TH}\backslash \left\{ i \right\}$, and ${{n}_{2,k}}=0$ for $k\notin \vartheta _{n}^{TH}$ and $k=i$.
\end{lemma}

\begin{IEEEproof}
When $l_{i}^{n}\ge l_{S}^{th}$, we denote ${{Z}_{i}}={{\left| {{h}_{{{S}_{i}}D}} \right|}^{2}}$, and ${{{Z}'}_{i}}=\underset{k\in \vartheta _{n}^{TH}\backslash \left\{ i \right\}}{\mathop{\max }}\,\left\{ {{Z}_{k}} \right\}_{k=1}^{K}$ with $i\in \vartheta _{n}^{TH}$, then $p_{i}^{n}$ is given as $p_{i}^{n}=\Pr \left( {{Z}_{i}}>{{{{Z}'}}_{i}} \right)$, which can be calculated as
\begin{align}\label{eq34}
p_{i}^{n}=\int_{0}^{\infty }{{{F}_{{{{{Z}'}}_{i}}}}\left( z \right){{f}_{{{Z}_{i}}}}\left( z \right)dz}.
\end{align}
We first derive
\begin{align}\label{eq35}
{{F}_{{{{{Z}'}}_{i}}}}\left( z \right)=\prod\limits_{k\in \vartheta _{n}^{TH}\backslash \left\{ i \right\}}{_{k=1}^{K}}\left( 1-{{e}^{-\frac{z}{{{{\bar{\gamma }}}_{{{S}_{k}}D}}}}} \right).
\end{align}

Referring to \cite{Yilmaz2013}, ${{F}_{{{{{Z}'}}_{i}}}}\left( z \right)$ can be rewritten as
\begin{align}\label{eq36}
{F_{{{Z'}_i}}}\left( z \right) &= \sum\limits_{{{\bf{n}}_1} \in \tau _{1,K}^i} {\prod\limits_{k = 1}^K {{{\left( { - 1} \right)}^{{n_{1,k}}}}{e^{ - \frac{{{n_{1,k}}}}{{{{\bar \gamma }_{{S_k}D}}}}z}}} } \notag\\
 &= \sum\limits_{{{\bf{n}}_1} \in \tau _{1,K}^i} {{{\left( { - 1} \right)}^{\sum\nolimits_{k = 1}^K {{n_{1,k}}} }}{e^{ - z\sum\nolimits_{k = 1}^K {\frac{{{n_{1,k}}}}{{{{\bar \gamma }_{{S_k}D}}}}} }}}.
\end{align}

In addition, we know that
\begin{align}\label{eq37}
{{f}_{{{Z}_{i}}}}\left( z \right)=\frac{1}{{{{\bar{\gamma }}}_{{{S}_{i}}D}}}{{e}^{-\frac{z}{{{{\bar{\gamma }}}_{{{S}_{i}}D}}}}}.
\end{align}

Substituting  \eqref{eq36} and \eqref{eq37} into \eqref{eq34}, and after some simple manipulations, \eqref{eq33} can be  readily derived.
\end{IEEEproof}

To derive the transition probability of the $k$th source from ${{\vec s}_{n}}$ to  ${{\vec s}_{{{n'}}}}$ in the IT mode, we present the following lemma.
\begin{lemma}\label{lemma2}
The CDF of the energy harvested at the $k$th source with the $n$th state
$\varepsilon _{k}^{n,{{i}^{*}}}$ is derived as
\begin{align}\label{eq27}
{F_{\varepsilon _k^{n,{i^ * }}}}\left( x \right)& = 1 - {e^{ - \frac{x}{{\eta {T_0}{P_B}{{\bar \gamma }_{B{S_k}}}}}}} - \frac{{{P_S}{{\bar \gamma }_{{S_{{i^ * }}}{S_k}}}}}{{{P_B}{{\bar \gamma }_{B{S_k}}} - {P_S}{{\bar \gamma }_{{S_{{i^ * }}}{S_k}}}}} \notag\\
&  \times \left( {{e^{{\mu _1}\frac{x}{{\eta {T_0}}}}} - 1} \right){e^{ - \frac{x}{{\eta {T_0}{P_S}{{\bar \gamma }_{{S_{{i^ * }}}{S_k}}}}}}},
\end{align}
where ${{\mu }_{1}}=\frac{{{P}_{B}}{{{\bar{\gamma }}}_{B{{S}_{k}}}}-{{P}_{S}}{{{\bar{\gamma }}}_{{{S}_{{{i}^{*}}}}{{S}_{k}}}}}{{{P}_{S}}{{{\bar{\gamma }}}_{{{S}_{{{i}^{*}}}}{{S}_{k}}}}{{P}_{B}}{{{\bar{\gamma }}}_{B{{S}_{k}}}}}$.

\end{lemma}

\begin{IEEEproof}
We first present the CDF of ${\varepsilon _k^{n,{i^ * }}}$ as
\begin{align}\label{eq31}
{F_{\varepsilon _k^{n,{i^ * }}}}\left( x \right) &= \Pr\left(\eta {{T}_{0}}\left( {{P}_{S}}{{\left| {{h}_{{{S}_{{{i}^{*}}}}{{S}_{k}}}} \right|}^{2}}+{{P}_{B}}{{\left| \mathbf{h}_{B{{S}_{k}}}^{T}\mathbf{w}_{2} \right|}^{2}} \right) \le x\right) \notag\\
&=\Pr({X}_{2}+{X}_{3}\le  \frac{x}{\eta {{T}_{0}}}  )\notag\\
&= \int_0^{\frac{x}{{\eta {T_0}}}} {{F_{{X_2}}}\left( {\frac{x}{{\eta {T_0}}} - y} \right){f_{{X_3}}}\left( y \right)dy},
\end{align}
where ${{X}_{2}}={{P}_{S}}{{\left| {{h}_{{{S}_{{{i}^{*}}}}{{S}_{k}}}} \right|}^{2}}$ and  ${{X}_{3}}\text{=}{{P}_{B}}{{\left| \sum\nolimits_{b=1}^{{{N}_{B}}}{{{h}_{{{B}_{b}}{{S}_{k}}}}w_{2,b}} \right|}^{2}}$. With $\mathbf{w}_{2}=\mathbf{h}_{BD}^{\dagger }$,  the construction of $\mathbf{w}_{2}$ is independent with ${{\mathbf{h}}_{B{{S}_{k}}}}$, so we have $\sum\nolimits_{b=1}^{{{N}_{B}}}{{{h}_{{{B}_{b}}{{S}_{k}}}}w_{2,b}}\sim \mathcal{C}\mathcal{N}\left( 0,{{{\bar{\gamma }}}_{B{{S}_{k}}}}\sum\nolimits_{b=1}^{{{N}_{B}}}{{{\left| w_{2,b} \right|}^{2}}} \right)$ with $\sum\nolimits_{b=1}^{{{N}_{B}}}{{{\left| w_{2,b} \right|}^{2}}}\text{=}1$ \cite{Ding2017}.
As a result, the PDF of ${{X}_{3}}$  is derived as
\begin{align}\label{eqfx3}
{{f}_{{{X}_{3}}}}\left( y \right)=\frac{1}{{{P}_{B}}{{{\bar{\gamma }}}_{B{{S}_{k}}}}}\exp \left( -\frac{y}{{{P}_{B}}{{{\bar{\gamma }}}_{B{{S}_{k}}}}} \right).
\end{align}

Besides, we know that ${{X}_{2}}\sim \exp \left( {{P}_{S}}{{{\bar{\gamma }}}_{{{S}_{{{i}^{*}}}}{{S}_{k}}}} \right)$, and its CDF is written as
\begin{align}\label{eq29}
{{F}_{{{X}_{2}}}}\left( x \right)=1-\exp \left( -\frac{x}{{{P}_{S}}{{{\bar{\gamma }}}_{{{S}_{{{i}^{*}}}}{{S}_{k}}}}} \right).
\end{align}

Substituting \eqref{eqfx3} and \eqref{eq29} into \eqref{eq31}, we prove \eqref{eq27}  in Lemma \ref{lemma2}.
\end{IEEEproof}

By utilizing the results in Lemmas \ref{lemma3} and \ref{lemma2},
the transition probability  from ${{\vec s}_{n}}$ to  ${{\vec s}_{{{n'}}}}$ in IT mode when  ${{\vec s}_{{{n}'}}}\in \Theta _{2}^{n,{n}',i}$ can be derived as
\begin{align}\label{eq39}
p_{{\vec s_n} \to {\vec s_{n'}}}^2 &= \sum\limits_{{{\bf{n}}_1} \in \tau _{1,K}^i} {\frac{{{{\left( { - 1} \right)}^{\sum\nolimits_{k = 1}^K {{n_{1,k}}} }}}}{{{{\bar \gamma }_{{S_i}D}}\sum\nolimits_{k = 1}^K {\frac{{{n_{1,k}}}}{{{{\bar \gamma }_{{S_k}D}}}}}  + 1}}} \notag\\
& \times \prod\limits_{\begin{array}{*{20}{c}}
{k \ne i,}\\
{k \in \vartheta _{n'}^{L}}
\end{array}} {_{k = 1}^K{\Lambda _1}\left( {i,k} \right)} \notag\\
& \times \prod\limits_{\begin{array}{*{20}{c}}
{k \ne i,}\\
{k \notin \vartheta _{n'}^{L}}
\end{array}} {_{k = 1}^K{\Lambda _2}\left( {i,k} \right)},
\end{align}
with
\begin{align}\label{eq40}
{\Lambda _1}\left( {i,k} \right) &= {e^{ - \frac{{{\varepsilon _{\Delta l_k^n}}}}{{\eta {T_0}{P_B}{{\bar \gamma }_{B{S_k}}}}}}} + \frac{{{P_S}{{\bar \gamma }_{{S_i}{S_k}}}}}{{{P_B}{{\bar \gamma }_{B{S_k}}} - {P_S}{{\bar \gamma }_{{S_i}{S_k}}}}} \notag\\
&  \times {e^{ - \frac{{{\varepsilon _{\Delta l_k^n}}}}{{\eta {T_0}{P_S}{{\bar \gamma }_{{S_i}{S_k}}}}}}}\left( {{e^{{\mu _1}\frac{{{\varepsilon _{\Delta l_k^n}}}}{{\eta {T_0}}}}} - 1} \right),
\end{align}
\begin{align}\label{eqLambda2k}
{\Lambda _2}\left( {i,k} \right) &= {\Lambda _1}\left( {i,k} \right) - \frac{{{P_S}{{\bar \gamma }_{{S_i}{S_k}}}{e^{ - \frac{{{\varepsilon _{\Delta l_k^n + 1}}}}{{\eta {T_0}{P_S}{{\bar \gamma }_{{S_i}{S_k}}}}}}}}}{{{P_B}{{\bar \gamma }_{B{S_k}}} - {P_S}{{\bar \gamma }_{{S_i}{S_k}}}}} \notag\\
 & \times \left( {{e^{{\mu _1}\frac{{{\varepsilon _{\Delta l_k^n + 1}}}}{{\eta {T_0}}}}} - 1} \right) - {e^{ - \frac{{{\varepsilon _{\Delta l_k^n + 1}}}}{{\eta {T_0}{P_B}{{\bar \gamma }_{B{S_k}}}}}}},
\end{align}
where $\vartheta _{{{n}'}}^{L}=\left\{ k:l_{k}^{{{n}'}}=L \right\}$ represents the set of sources whose energy level is $L$ at state ${{\vec s}_{{{n'}}}}$.


To conclude, the transition probability  from ${{\vec s}_{n}}$ to  ${{\vec s}_{{{n'}}}}$ is summarized as
\begin{align}\label{}
{{p}_{{{\vec s}_{n}}\to {{\vec s}_{{{n}'}}}}}=\left\{ \begin{matrix}
   p_{{{\vec s}_{n}}\to {{\vec s}_{{{n}'}}}}^{1}, & {{\vec s}_{n}}\in {{\Theta }_{1}},{{\vec s}_{{{n}'}}}\in \Theta _{1}^{n,{n}'},  \\
   p_{{{\vec s}_{n}}\to {{\vec s}_{{{n}'}}}}^{2}, & {{\vec s}_{n}}\in {{\Theta }_{2}},{{\vec s}_{{{n}'}}}\in \Theta _{2}^{n,{n}',i},  \\
   0, & others.  \\
\end{matrix} \right.
\end{align}
Let us denote $\mathbf{A}\in {{\mathbb{R}}^{N\times N}}$ as the state transition matrix of the proposed network, where the $\left( n,{n}' \right)$-th element ${{a}_{n,{n}'}}$ represents the probability to transfer from ${{\vec s}_{n}}$ to  ${{\vec s}_{{{n'}}}}$, and is given by
\begin{align}\label{}
{{a}_{n,{n}'}}={{p}_{{{\vec s}_{n}}\to {{\vec s}_{{{n}'}}}}}.
\end{align}

We then formulate the stationary distribution $\bm{\pi} \in {{\mathbb{R}}^{N\times 1}}$ for the energy states, where its $n$th element, ${{\bm{\pi} }_{n}}$, stands for the stationary probability of state ${{\vec s}_{n}}$ for the network. It is easily to know that $\mathbf{A}$ is irreducible and row stochastic. As a consequence, a unique stationary distribution must exist that satisfies \cite{BiandChen2016,Krikidis2012}
\begin{align}\label{eq49}
\bm{\pi} ={{\mathbf{A}}^{T}}\bm{\pi} .
\end{align}
According to \cite[Eq. (12)]{Krikidis2012}, the solve of \eqref{eq49} could be derived as
\begin{align}\label{pai}
\bm{\pi} ={{\left( {{\mathbf{A}}^{T}}-\mathbf{E}+\mathbf{Q} \right)}^{-1}}\mathbf{b},
\end{align}
where $\mathbf{b}={{\left( 1,1,\cdots 1 \right)}^{T}}$, $\mathbf{E}$ is the identity matrix, and $\mathbf Q$ is an all-ones matrix.

For a better comprehension, Fig. \ref{figA} depicts the block diagram for the construction
of the state transition matrix $\textbf{A}$ based on the system parameters $K$, $L$ and $l_S$. Fig. \ref{figeg} illustrates the transitions of the states for  a simple  example with $K=2$ and $L=2$. The corresponding  state transition matrix $\textbf{A}$ could be derived as \eqref{transA}. By applying the described approach, the stationary state probabilities of all the states, as well as the EOP, COP and ATD for $S_1$ and $S_2$ are obtained as shown in Table \ref{tbeg}. The related parameters are set as $P_B=30$ dBm, $\varepsilon_T=20$ mJ, $\varepsilon_{S}^{th}=10$ mJ, $l_S=l_S^{th}$, $\eta=0.8$, $R_t=3$ bits/s/Hz, $x_B=-3$ m, $x_D=200$ m, $r_S=1$ m, $N_B=5$, $N_0=-80$ dBm, and $\alpha=3$. The coordinates of $B$ and $D$ as well as $S_1$ and $S_2$ are $B=(-3 ,0)$, $D=(200,0)$, $S_1=(-1,0)$ and $S_2=(0,1)$, respectively. The unit of distance is the meter.


 \begin{figure}
 \begin{center}
   \includegraphics[width=3.5in,angle=0]{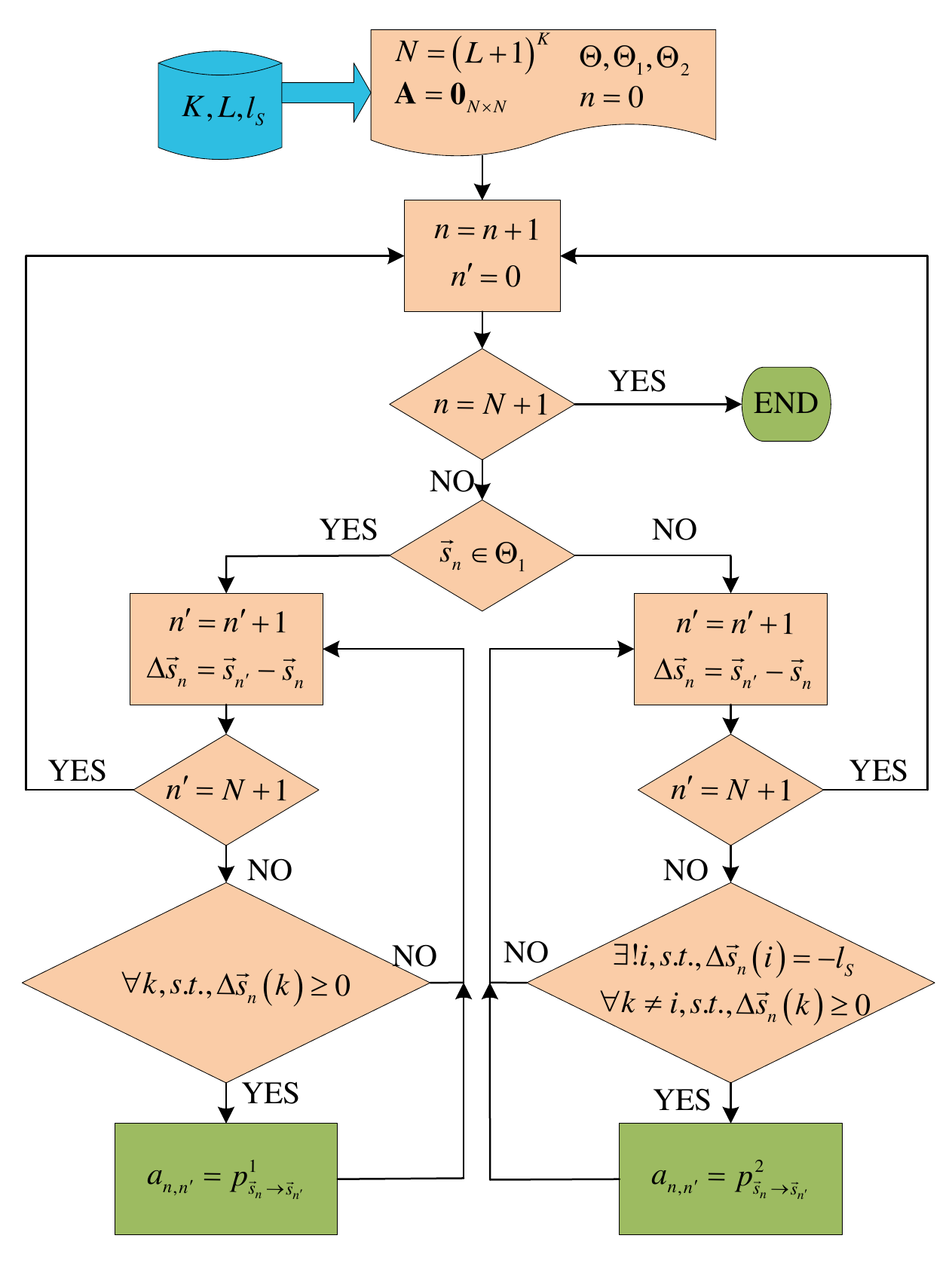}\\
   \caption{Flow diagram for the generation of the state transition matrix \bf{A}.}\label{figA}
 \end{center}
 \end{figure}

 \begin{figure}
 \begin{center}
   \includegraphics[width=3.5in,angle=0]{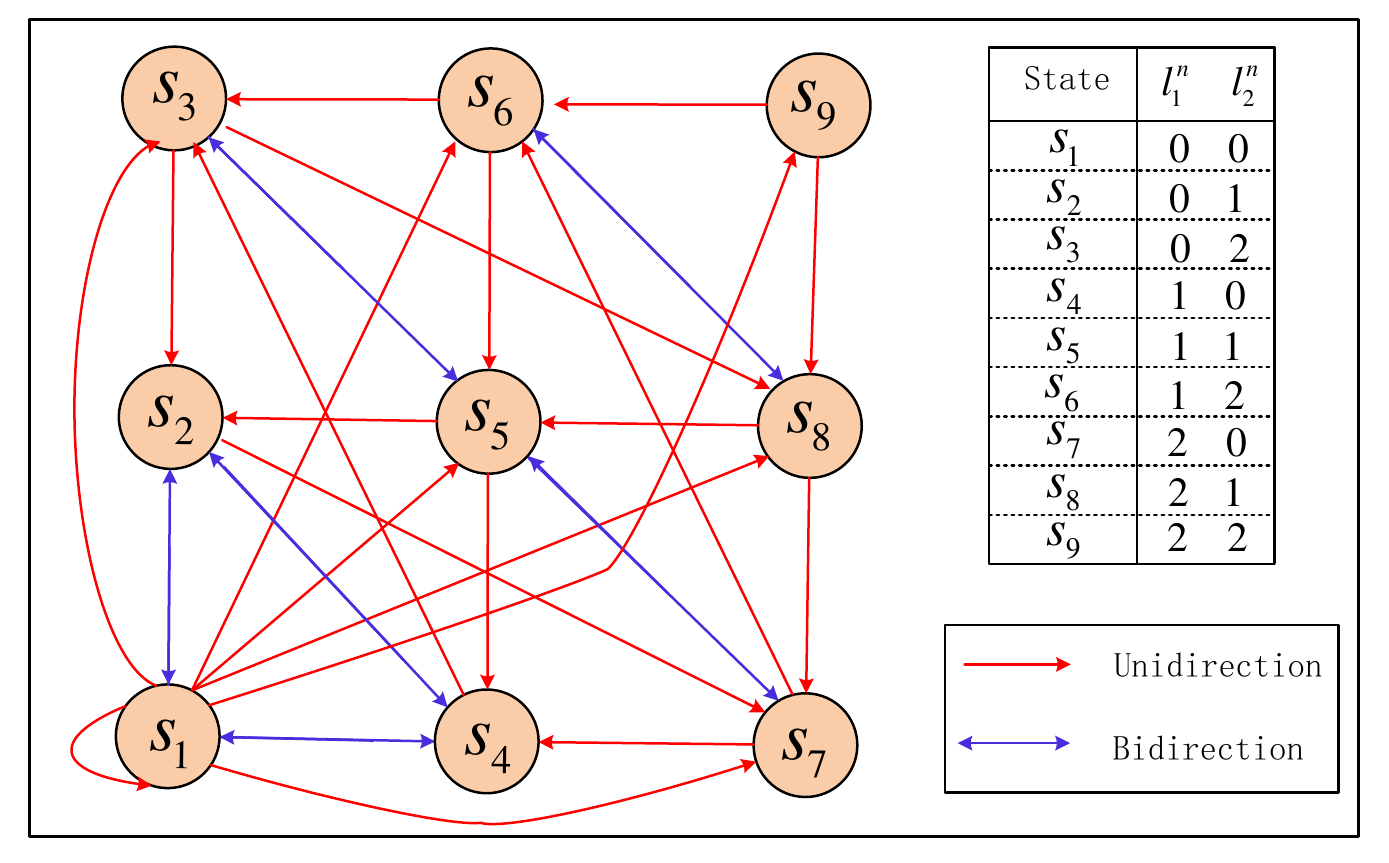}\\
   \caption{State diagram of the Markov chain representing the states of the storages and the transitions between them for a case with $K=2$ and $L=2$.}\label{figeg}
 \end{center}
 \end{figure}

\begin{figure*}[bt!]
\rule{\linewidth}{1pt}
\begin{equation}\label{transA}
\textbf{A} = \left( {\begin{array}{*{20}{c}}
{{\rm{0}}{\rm{.0356}}}&\!\!\!{{\rm{0}}{\rm{.0237}}}&{{\rm{0}}{\rm{.0471}}}&{{\rm{0}}{\rm{.0318}}}&{{\rm{0}}{\rm{.0212}}}&{{\rm{0}}{\rm{.0421}}}&{{\rm{0}}{\rm{.2674}}}&{{\rm{0}}{\rm{.1779}}}&{{\rm{0}}{\rm{.3533}}}\\
{{\rm{0}}{\rm{.0851}}}&\!\!\!0&0&{{\rm{0}}{\rm{.0972}}}&0&0&{{\rm{0}}{\rm{.8177}}}&0&0\\
0&\!\!\!{{\rm{0}}{\rm{.0851}}}&\!\!\!0&0&{{\rm{0}}{\rm{.0972}}}&0&0&{{\rm{0}}{\rm{.817}}7}&0\\
{{\rm{0}}{\rm{.2737}}}&\!\!\!{{\rm{0}}{\rm{.2427}}}&{{\rm{0}}{\rm{.4836}}}&0&0&0&0&0&0\\
0&\!\!\!{{\rm{0}}{\rm{.1358}}}&\!\!\!{{\rm{0}}{\rm{.3604}}}&{{\rm{0}}{\rm{.0429}}}&0&0&{{\rm{0}}{\rm{.4609}}}&0&0\\
0&\!\!\!{\rm{0}}&\!\!\!{{\rm{0}}{\rm{.4963}}}&0&{{\rm{0}}{\rm{.0429}}}&0&0&{{\rm{0}}{\rm{.4609}}}&0\\
0&\!\!\!0&0&{{\rm{0}}{\rm{.2737}}}&{{\rm{0}}{\rm{.2427}}}&{{\rm{0}}{\rm{.4836}}}&0&0&0\\
0&\!\!\!0&0&0&{{\rm{0}}{\rm{.1358}}}&{{\rm{0}}{\rm{.3604}}}&{{\rm{0}}{\rm{.5037}}}&0&0\\
0&\!\!\!0&0&0&0&{{\rm{0}}{\rm{.4963}}}&0&{{\rm{0}}{\rm{.5037}}}&0
\end{array}} \right).
\end{equation}
\setcounter{equation}{38}
\rule{\linewidth}{1pt}
\end{figure*}

\begin{table}[htbp]
  \centering
  \caption{Illustration of example when $K=2$, $L=2$.}
    \begin{tabular}{|c|c|c|c|c|}
    \toprule
    State $s_n$ & [$l_1^n$,$l_2^n$] & $\bm{\pi}_n$ & [$p_1^n$,$P_2^n$] & State COP \\
    \midrule
    $s_1$ & [0,0] & 0.0220 & [0,0] & 0 \\
    $s_2$ & [0,1] & 0.0434 & [0,1] & 0.0545 \\
    $s_3$ & [0,2] & 0.1587 & [0,1] & 0.0545 \\
    $s_4$ & [1,0] & 0.0640 & [1,0] & 0.0553 \\
    $s_5$ & [1,1] & 0.1022 & [0.4963,0.5037] & 0.0030 \\
    $s_6$ & [1,2] & 0.1811 & [0.4963,0.5037] & 0.0030 \\
    $s_7$ & [2,0] & 0.1998 & [1,0] & 0.0553 \\
    $s_8$ & [2,1] & 0.2210 & [0.4963,0.5037] & 0.0030 \\
    $s_9$ & [2,2] & 0.0078 & [0.4963,0.5037] & 0.0030 \\
    \midrule
    \multicolumn{2}{|c|}{\multirow{4}[8]{*}{Derived results}} & EOP   & [$p_{T,1}$,$p_{T,2}$] & Overall COP \\
\cmidrule{3-5}    \multicolumn{2}{|c|}{} & \multirow{3}[6]{*}{0.0220} & [0.5179,0.4601] & \multirow{3}[6]{*}{0.0271} \\
\cmidrule{4-4}    \multicolumn{2}{|c|}{} &       & [$\bar T_1$,$\bar T_2$] &  \\
\cmidrule{4-4}    \multicolumn{2}{|c|}{} &       & [1.9309,2.1734] &  \\
    \bottomrule
    \end{tabular}%
  \label{tbeg}%
\end{table}%


\section{Outage and Delay}\label{sec5}
In this section, we characterize the performance in
terms of outage and delay. Specifically, we focus on the derivations for
 the EOP in the non-IT mode,  the COP, and the average transmission delay (ATD) in the IT model. To reveal key insights of the proposed network, we derive exact expressions for the EOP, COP, and ATD of proposed networks.

\subsection{Energy Outage Probability}

In the proposed network,  the EOP is defined as the network energy outage  in the non-IT mode when  all the sources experience energy outage. The EOP is derived as in the following theorem.
\begin{theorem}
The EOP for the multi-source WPT network is derived as
\begin{align}\label{eqeop}
{{P}_{EO}}=\sum\limits_{{{\vec s}_{n}}\in {{\Theta }_{1}}}{{{\bm{\pi} }_{n}}},
\end{align}
where  $\bm{\pi}_n$ has been derived as in \eqref{pai}.
\end{theorem}
\begin{proof}
According to the definition of the network energy outage given in \ref{sec3}, the EOP of the proposed network is readily derived.
\end{proof}

\begin{corollary}
The EOP for the multi-source WPT network when
the transmit power of the PB goes to infinity (${P_B} \to \infty$) is given by
\begin{align}
{P_{EO}} = \left\{ {\begin{array}{*{20}{c}}
{\sum\limits_{{{\vec s}_n} \in {\Theta _1}} {{\bm{\pi} _n}} ,}&{K = 1,}\\
{0,}&{K \ge 2,}
\end{array}} \right.
\end{align}
where  $\bm{\pi}_n$ has been derived as in \eqref{pai}.
\end{corollary}

\begin{proof}
We fisrt denote $\mathbf{A}=\left( \begin{matrix}
   {{\mathbf{A}}_{1,1}} & {{\mathbf{A}}_{1,2}}  \\
   {{\mathbf{A}}_{2,1}} & {{\mathbf{A}}_{2,2}}  \\
\end{matrix} \right) $, where  	${{\mathbf{A}}_{i,j}}$ represents the transition matrix from ${{\Theta }_{i}}$
to ${{\Theta }_{j}}$, $i,j\in \left\{ 1,2 \right\}$. We also denote the stationary distributions of states in $\Theta_1$ and   $\Theta_2$  as $\bm{\pi}_{1}^{\Theta} \in \mathbb{R}^{N_1 \times 1}$ and $\bm{\pi}_{2}^{\Theta}\in \mathbb{R}^{N_2 \times 1}$, respectively.

When ${P_B} \to \infty$, we have $\varepsilon _{k}^{n}\to \infty $, $\varepsilon _{k}^{n,{{i}^{*}}}\to \infty $. Therefore, if $\vec s_n \in \Theta_1$, it will always transfer to the all-full state $[L,\cdots,L]$, because the harvested energy at each source will always exceed the energy capacity. Similarly, if $\vec s_n \in \Theta_2$, it will always transfer to an almost-all-full state $[L,\cdots,l_{i^*},\cdots,L]$, where $0 \le l_{i^*} \le L-l_S$ and $i^* \in \{1, \cdots, K\}$ is the index of the selected IT source. Note that both $[L,\cdots, L]$ and $[L,\cdots,l_{i^*},\cdots,L]$ are an element of $\Theta_2$. As a result, for $K\ge 2$, regardless of any current state the network remains, it will never transfer to  $\Theta_1$, which results to ${{\mathbf{A}}_{1,1}}={{\mathbf{0}}_{{{N}_{1}}\times {{N}_{1}}}}$ and ${{\mathbf{A}}_{2,1}}={{\mathbf{0}}_{{{N}_{2}}\times {{N}_{1}}}}$.
Substituting the derived results into \eqref{eq49}, we derive the matrix-form equation as
\begin{equation}\label{}
{\left( {\begin{array}{*{20}{c}}
{{{\bf{0}}_{{N_1} \times {N_1}}}}&{{{\bf{A}}_{1,2}}}\\
{{{\bf{0}}_{{N_2} \times {N_1}}}}&{{{\bf{A}}_{2,2}}}
\end{array}} \right)^T}\left( {\begin{array}{*{20}{c}}
{{\bm{\pi}_{1}^{\Theta}}}\\
{{\bm{\pi}_{2}^{\Theta}}}
\end{array}} \right) = \left( {\begin{array}{*{20}{c}}
{{\bm{\pi}_{1}^{\Theta}}}\\
{{\bm{\pi}_{2}^{\Theta}}}
\end{array}} \right),
\end{equation}
which yields to $\bm{\pi}_{1}^{\Theta}=\bf{0}$. Referring to \eqref{eqeop}, we derive the EOP of the network as
\begin{equation}\label{}
{P_{EO}}\mathop  =  \sum\limits_{{\vec s_n} \in {\Theta _1}} \bm{\pi}_n  = \sum\limits_{k = 1}^{{N_1}} {\bm{\pi} _{1,k}^\Theta }  = 0,
\end{equation}
where ${\bm{\pi} _{1,k}^\Theta }$ denotes the $k$th element of ${\bm{\pi} _{1}^\Theta}$.

Whereas, for $K=1$, due to the half-duplex nature, the single source can not harvest energy when it transmits information. Hence, if the network remains in $\Theta_2$, the source will always consumed energy until an energy outage event occurs. As such, the EOP when $K=1$ is derived using \eqref{eqeop}.
\end{proof}

\subsection{Connection Outage Probability}
The COP quantifies the probability that the information can not be correctly decoded at the legitimate receiver when the IT operation actually takes place.
According to the total probability theorem, and considering the fact that no data is transmitted when energy outage occurs.
\begin{theorem}
The overall COP for the multi-source WPT network is derived as
\begin{align}
P_{CO}=\sum\limits_{{{\vec s}_{n}}\in {{\Theta_2 }}}{{{\bm{\pi} }_{n}}}\prod\limits_{k \in \vartheta _n^{TH}} {_{k = 1}^K\left( {1 - \exp \left( { - \frac{{\sigma _D^2\gamma _{th}^t}}{{{{\bar \gamma }_{{S_k}D}}{P_S}}}} \right)} \right)},
\end{align}
where  $\bm{\pi}_n$ has been derived as in \eqref{pai}.
\end{theorem}

\begin{proof}
According to the total probability theorem, the overall COP of the multi-source WPT network can be calculated as
\begin{align}\label{eq52}
P_{CO}=\sum\limits_{{{\vec s}_{n}}\in {{\Theta }}}{{{\bm{\pi} }_{n}}}P_{CO,n} \mathop  = \limits^{\left( a \right)} \sum\limits_{{{\vec s}_{n}}\in {{\Theta }_{2}}}{{{\bm{\pi} }_{n}}}P_{CO,n}.
\end{align}
where the result after $\mathop  = \limits^{\left( a \right)}$ is derived according to the fact that no information would be transmitted  when energy outage occurs, and  $P_{CO,n}$ represents the COP  when the network remains at state ${{\vec s}_{n}}$,  which is derived as
\begin{align}\label{eq53}
P_{CO,n}&=\sum\limits_{i\in \vartheta _{n}^{TH}}{\Pr \left( {{i}^{*}}=i \right)}\Pr \left( \gamma _{D}^{n,{{i}^{*}}}<\gamma _{th}^{t} \right).\notag\\
&=\sum\limits_{i\in \vartheta _{n}^{TH}}{p_{i}^{n}}{{F}_{\gamma _{D}^{n,{{i}^{*}}}}}\left( \gamma _{th}^{t} \right)
\end{align}
where $\gamma _{th}^{t}={{2}^{{{R}_{t}}}}-1$, ${{R}_{t}}$ (bits/s/Hz) denotes the transmission rate of the network.

According to the selection policy described in \eqref{eq6}, we can present the CDF of ${\gamma _{D}^{n,{{i}^{*}}}}$ as
\begin{align}\label{eq54}
{{F}_{\gamma _{D}^{n,{{i}^{*}}}}}\left( x \right)={{F}_{{{Y}_{SD}}}}\left( \sigma _{D}^{2}x \right),
\end{align}
where ${{Y}_{SD}}={{P}_{S}}\underset{k\in \vartheta _{n}^{TH}}{\mathop{\max }}\,\left\{ {{\left| {{h}_{{{S}_{k}}D}} \right|}^{2}} \right\}$. After some  manipulations, the CDF of $Y_{SD}$ is derived as
\begin{align}\label{eq55}
{{F}_{{{Y}_{SD}}}}\left( y \right)=\prod\limits_{k\in \vartheta _{n}^{TH}}{_{k=1}^{K}\left( 1-\exp \left( -\frac{y}{{{{\bar{\gamma }}}_{{{S}_{k}}D}}{{P}_{S}}} \right) \right)}.
\end{align}
It is easy to find from \eqref{eq54} and \eqref{eq55} that the CDF of ${{\gamma _{D}^{n,{{i}^{*}}}}}$ has no relationship with the selection probabilities of every  source $S_i$, $i \in \vartheta _{n}^{TH}$. Besides, according to the total probability theorem, we derive $\sum\limits_{i\in \vartheta _{n}^{TH}}p_i^n=1$. Hence, $P_{CO,n}$  can be calculated as

\begin{align}\label{pcon}
P_{CO,n}&={{F}_{\gamma _{D}^{n,{{i}^{*}}}}}\left( \gamma _{th}^{t} \right)\notag\\
&=\prod\limits_{k \in \vartheta _n^{TH}} {_{k = 1}^K\left( {1 - \exp \left( { - \frac{{\sigma _D^2\gamma _{th}^t}}{{{{\bar \gamma }_{{S_k}D}}{P_S}}}} \right)} \right)}.
\end{align}

%
%

By substituting ${{\bm{\pi }}_{n}}$ in \eqref{pai} and $P_{CO,n}$ in \eqref{pcon} into \eqref{eq52}, the COP of the proposed network is derived.
\end{proof}

\subsection{Average Transmission Delay}
IN the IT mode,  there would be at most one source to send messages at each time slot, a transmission delay is caused at each source. In practical, we may concern that how many time slots on average a specific source need to wait for to be selected for IT operation, which can be quantified by the average transmission delay (ATD).

Before delving into the investigation, we will clarify the fundamental conception of ATD by giving out a simple example. Let us start by looking into  the network of $K$ energy-sufficient sources, where all the sources can be selected for IT operation equally. It is readily known that at each time slot, each source has the transmision probability of $1/K$. In other words, for each source, a time slot is allocated once on average within $K$ slots. As a consequence,  the ATD would be $\bar T=KT_0$ for every source in this network \footnote{For the extreme case of $K=1$, we can find that the source can always transmit successively. We say that the ATD of this network is $\bar T=T_0$, even though no time slot is needed to wait for IT operation.}.
However, in our proposed energy storage networks, whether a specific source can be selected for IT operation differs for different storage states. In other words, the transmission probability of a certain source is not fixed, and all the sources do not have the equal transmission probability as well.

In order to solve this problem, we denote $p_{T,i}^n$ as the transmission probability for source $S_i$ at state $s_n$. According to the previous description, $p_{T,i}^n$ can be derived as
\begin{equation}\label{ptin}
p_{T,i}^n = \left\{ {\begin{array}{*{20}{c}}
{0,}&{l_i^n < {l_S},}\\
{p_i^n,}&{l_i^n \ge {l_S}.}
\end{array}} \right.
\end{equation}

\begin{theorem}
The average transmission probability for source $S_i$ and its ATD are  derived as
%
\begin{equation}\label{}
{{p}_{T,i}}=\sum\limits_{{{\vec s}_{n}}\in \Theta }{{{\bm{\pi} }_{n}}p_{T,i}^{n}}=\sum\limits_{{{\vec s}_{n}}\in {{\Theta }_{2}}}{{{\bm{\pi} }_{n}}p_{T,i}^{n}},
\end{equation}
and
\begin{equation}\label{}
{\bar T_i} = \frac{{{{T}_0}}}{{{p_{T,i}}}} = \frac{{{{T}_0}}}{{\sum\limits_{{\vec s_n} \in {\Theta _2}} {{\bm{\pi} _n}p_{t,i}^n} }},
\end{equation}
respectively, where  $\bm{\pi}_n$ has been derived as in \eqref{pai}.
\end{theorem}

\begin{proof}
The proof is omitted.
\end{proof}

\section{Numerical Results}\label{sec6 NUMERICAL RESULTS}
In this section, we present the numerical results  to illustrate the
impacts of various system parameters on the  performance of the  proposed network. As shown in the below  figures, the theoretical results are in exact
agreement with the numerical simulations, which show  the
correctness of the analysis. Without any loss of generality, all the nodes are set in a two-dimensional plane in all simulations, and the coordinates of $B$ and $D$ are set as $B=(x_B,0)$ and $D=(x_D,0)$, and the coordinates of the source nodes are assumed to be $S_1=(-r_S,0)$, $S_2=(0,r_S)$, $S_3=(r_S,0)$, $S_4=(0,-r_S)$, $S_5=({\textstyle{{\sqrt 2 } \over 2}} r_S,{\textstyle{{\sqrt 2 } \over 2}} r_S)$ and $S_6=(-{\textstyle{{\sqrt 2 } \over 2}} r_S,-{\textstyle{{\sqrt 2 } \over 2}} r_S)$, respectively. With $K$ sources, we take sources from $S_1$ to $S_K$ in order automatically, and we set $\alpha=3$, $R_t=3$ bits/s/Hz, $N_B=5$, $N_0=-80$ dBm, and $\varepsilon_S^{th}=10$ mJ.

\begin{figure}
\begin{center}
  \includegraphics[width=3.5in,angle=0]{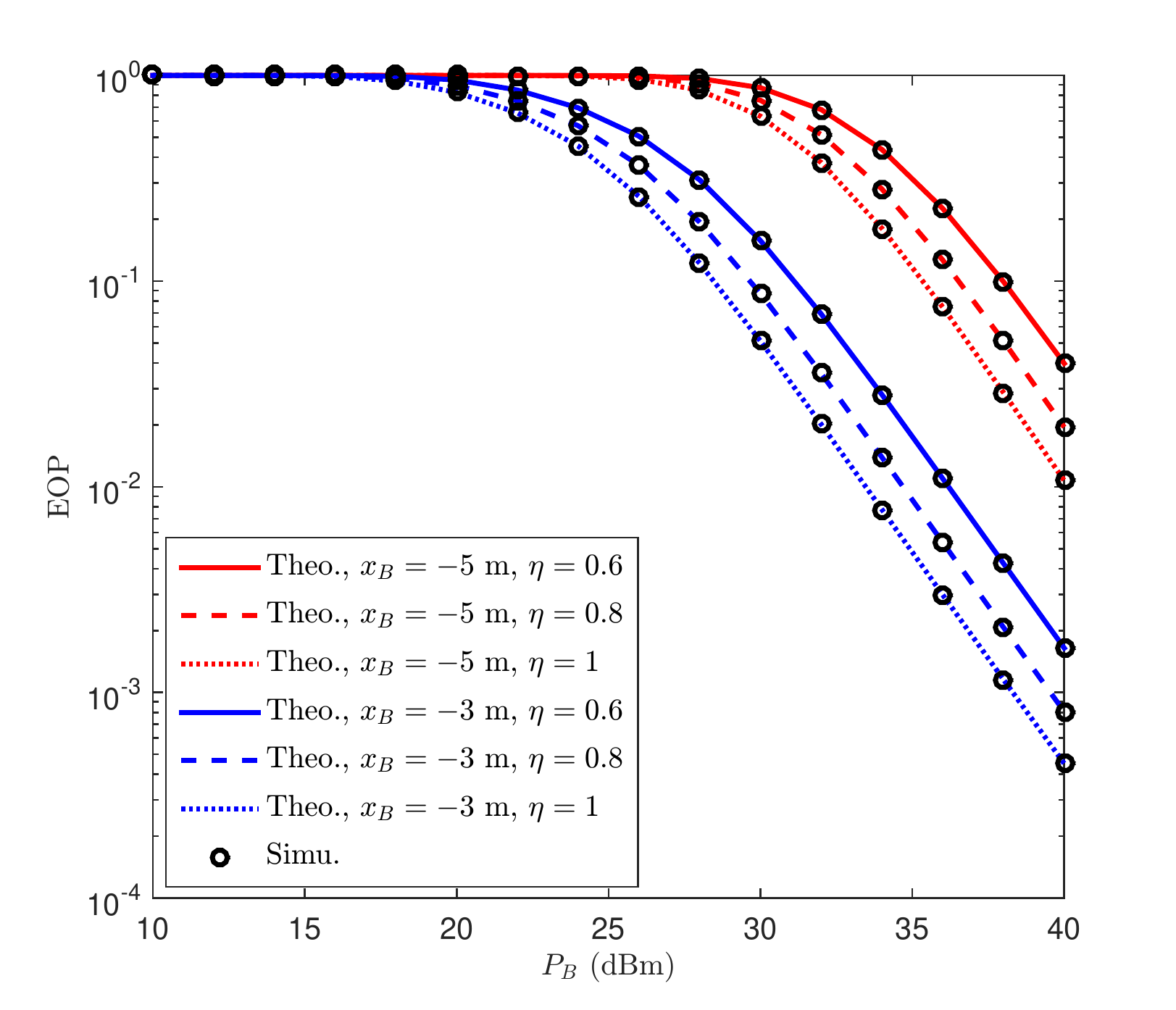}\\
  \caption{EOP of the multi-source WPT network versus the transmit power of power beacon $P_B$  with different $x_B$ and $\eta$. $P_B=30$ dBm, $l_S=l_S^{th}$, $\varepsilon_T=40$mJ, $x_D=200$ m, and $r_S=1$ m.}\label{figXvecPB_dBS_eta_change}
\end{center}
\end{figure}

Fig.  \ref{figXvecPB_dBS_eta_change} plots the EOP of the multi-source WPT network versus the transmit power of the power beacon $P_B$ with different $x_B$ and $\eta$.  As can be seen from this  figures, the EOP declines rapidly when $P_B$ increases.
Besides, it shows that the EOP will grow severely when $x_B$ increases. Moreover, the EOP  also raises when $\eta$ becomes smaller. This can be well understood because a greater $x_B$ implies a farther distance of energy transmission, which results in the decline of  accessible energy that can be harvested by sources due to a much severer path poss. Likewise, a smaller $\eta$ means a lower energy efficiency, which indicates that less energy can be converted by sources and saved in their storages.

\begin{figure}
\begin{center}
  \includegraphics[width=3.5in,angle=0]{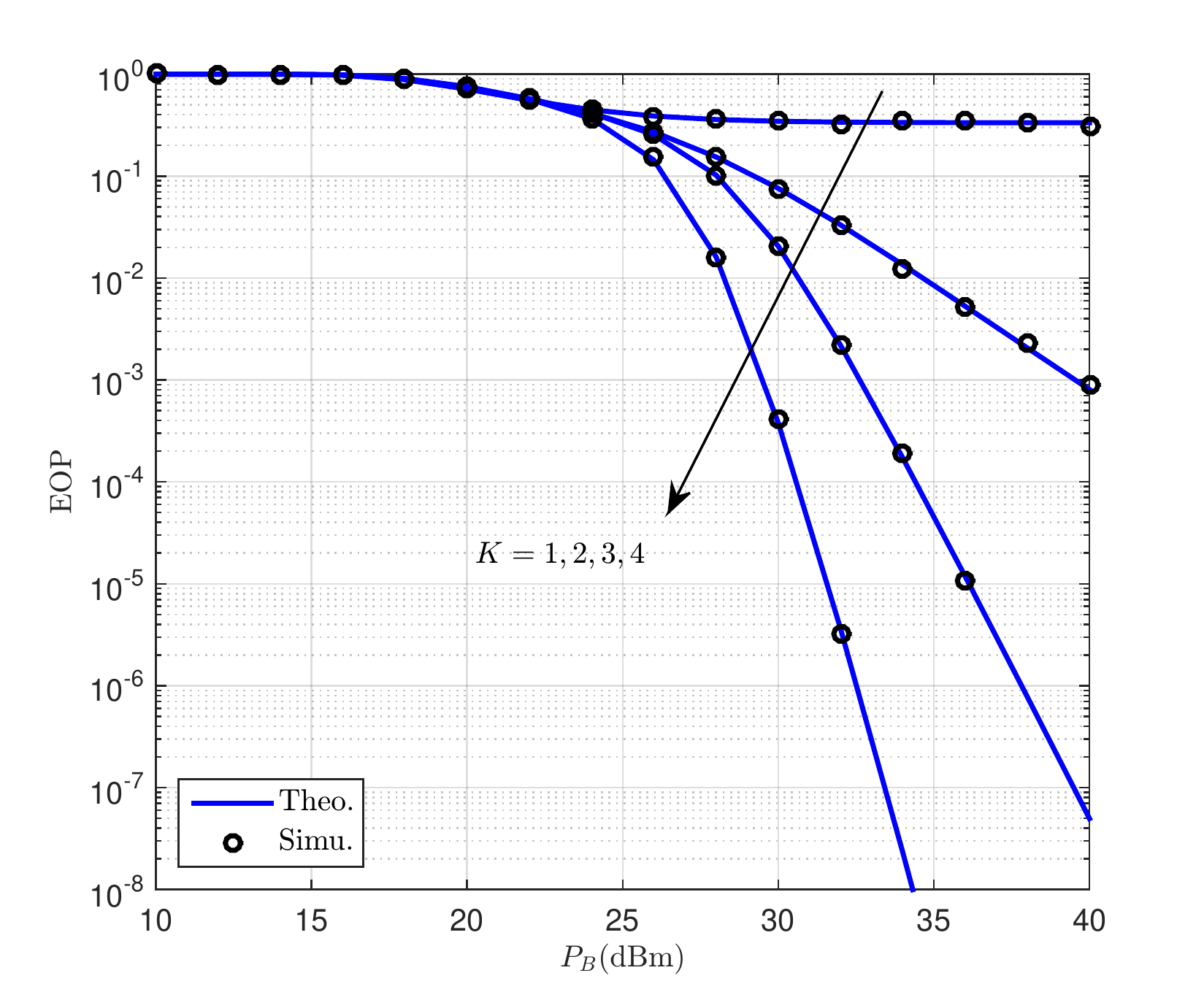}\\
  \caption{EOP of the multi-source WPT network versus the transmit power of  power beacon $P_B$ with different $K$. $l_S=l_S^{th}$, $\eta=0.8$, $\varepsilon_T=40$ mJ, $L=2$, $x_B=-3$ m, $x_D=200$ m, and $r_S=1$ m.}\label{figvecPB_M}
\end{center}
\end{figure}


Figs.  \ref{figvecPB_M} plots the EOP of the multi-source WPT network versus the transmit power of the power beacon $P_B$ with different  $K$. It is depicted that, the EOP performance
is rather poor when $K=1$, which however can be greatly improved when multiple sources are deployed, especially when a larger $P_B$ can be provided. This finding is of significant importance because it indicates the effectiveness to greatly decrease the EOP of network by increasing the number of the sources.

\begin{figure}
\begin{center}
  \includegraphics[width=3.5in,angle=0]{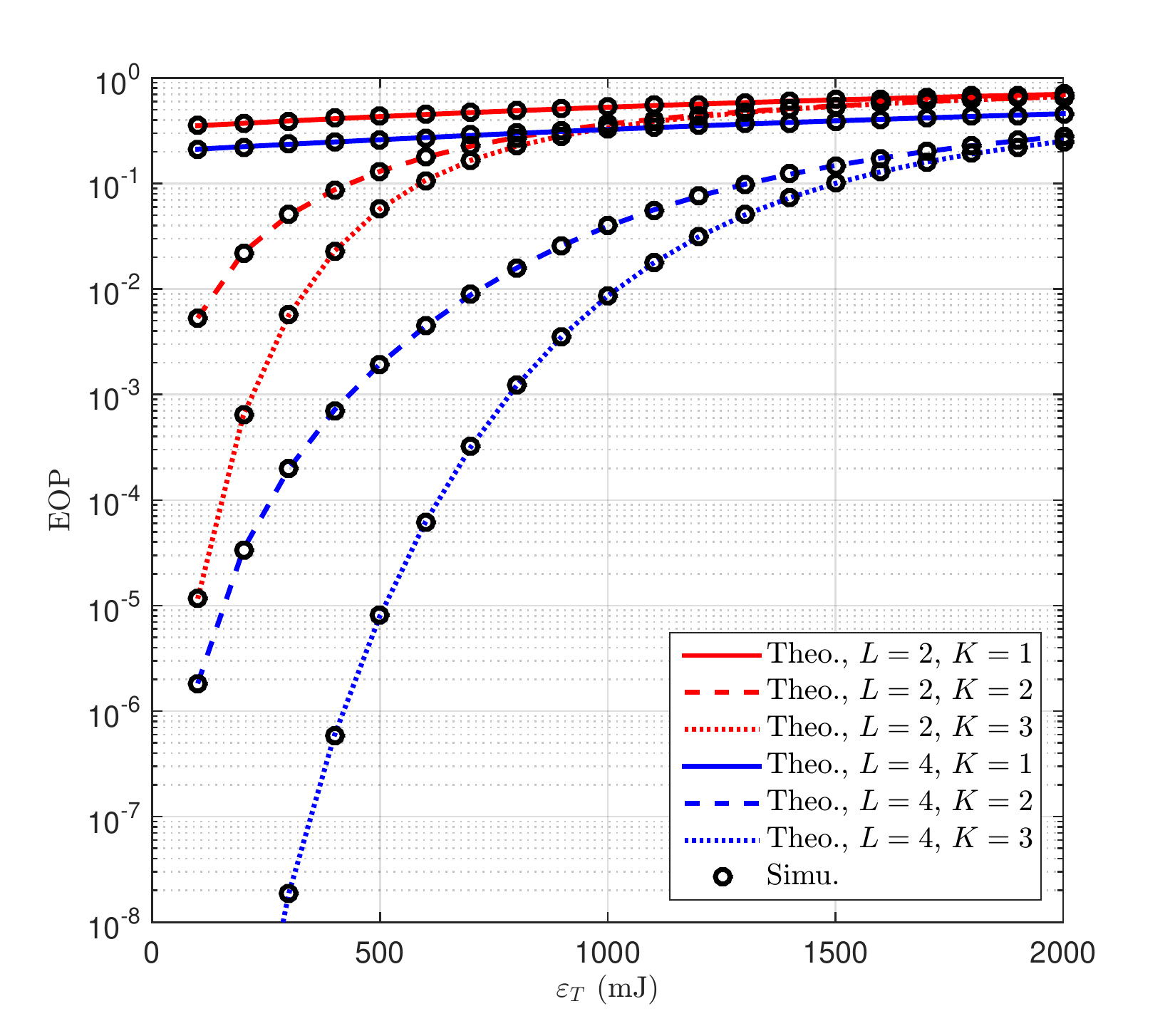}\\
  \caption{EOP of the multi-source WPT network versus the energy capacity $\varepsilon_{T}$ with different $K$ and $L$. $P_B=30$ dBm, $l_S=l_S^{th}$, $\eta=0.8$, $x_B=-3$ m, $x_D=200$ m, and $r_S=1$ m.}\label{figXvecenergyT_M_L_change}
\end{center}
\end{figure}

\begin{figure}
\begin{center}
  \includegraphics[width=3.5in,angle=0]{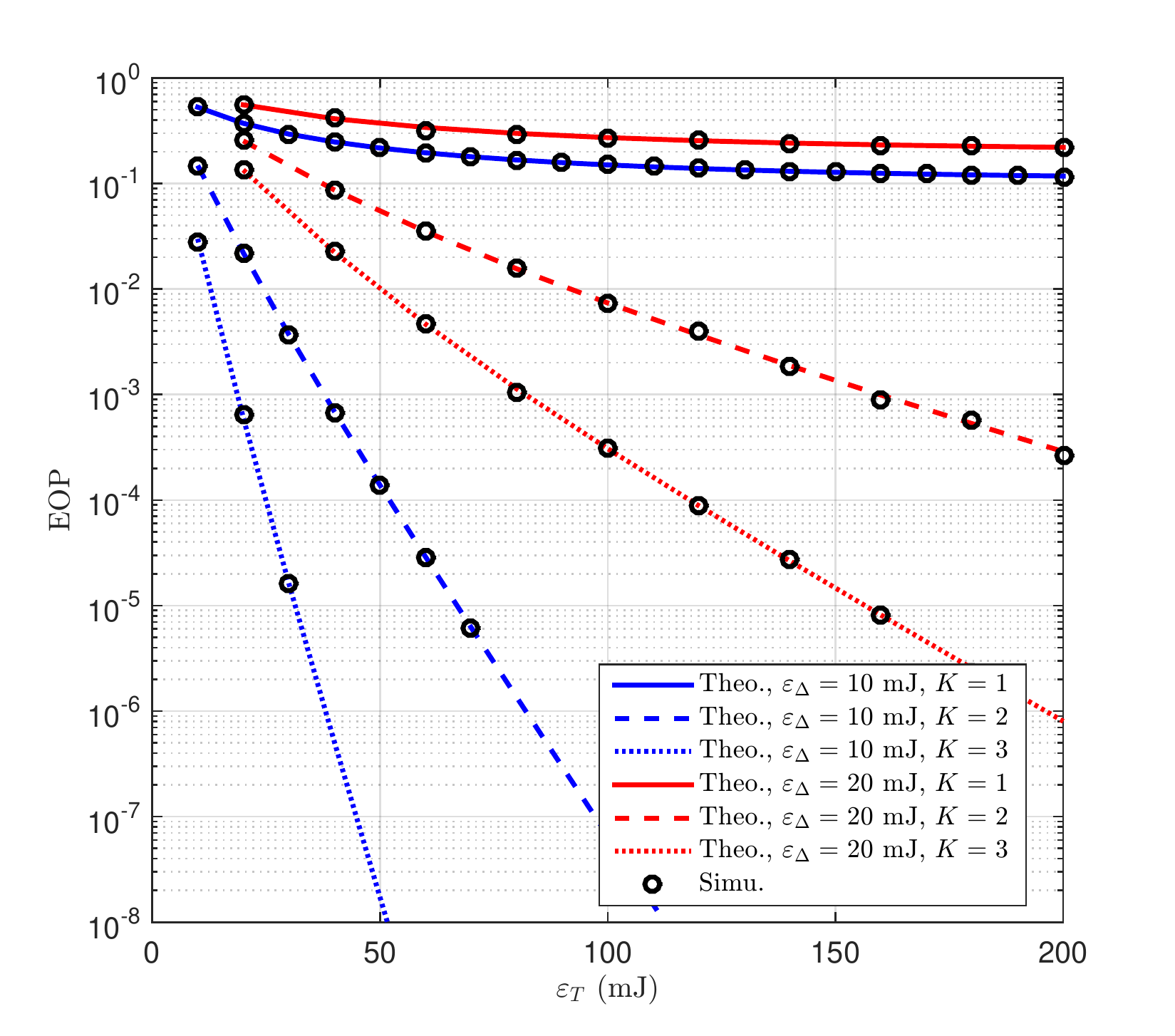}\\
  \caption{EOP of the multi-source WPT network versus the energy capacity $\varepsilon_{T}$ with different $K$ and $\varepsilon_{\Delta}$. $P_B=30$ dBm, $l_S=l_S^{th}$, $\eta=0.8$, $x_B=-3$ m, $x_D=200$ m, and $r_S=1$ m.}\label{figXvecenergyT_M_deltaenergy_change}
\end{center}
\end{figure}

Figs. \ref{figXvecenergyT_M_L_change}-\ref{figXvecenergyT_M_deltaenergy_change} examine the EOP of the multi-source WPT network versus the energy capacity $\varepsilon_{T}$.
We note that in Fig. \ref{figXvecenergyT_M_L_change}, the discretization level of the network $L$ is fixed so that the single unit of energy ${\varepsilon _\Delta }$ grows proportionally with the growth of ${\varepsilon _T}$. However, in Fig. \ref{figXvecenergyT_M_deltaenergy_change}, $L$ and ${\varepsilon _T}$ are proportionally increased  while keeping ${\varepsilon _\Delta}$ unchanged. From Fig. \ref{figXvecenergyT_M_L_change}, we find that for a specific $\varepsilon_{T}$, the EOP reduces when a larger $L$ is applied. By contrast, for a given $L$, the EOP grows rapidly with the increase of $\varepsilon_{T}$. Specifically, when $\varepsilon_{T}$ is large enough, the EOP approaches close to 1, even when multiple sources are applied. On the contrary, it is demonstrated  in  Fig. \ref{figXvecenergyT_M_deltaenergy_change} that, the EOP declines significantly with the increase of $\varepsilon_{T}$, which differs from the behavior shown in Fig. \ref{figXvecenergyT_M_L_change}. We note that the harvested energy at each time slot is limited.
Therefore, less energy could be harvested for the network when the  single unit of energy  ${\varepsilon _\Delta }$ grows, as a larger ${\varepsilon _\Delta }$ is more difficult to be satisfied.

\begin{figure}
\begin{center}
  \includegraphics[width=3.5in,angle=0]{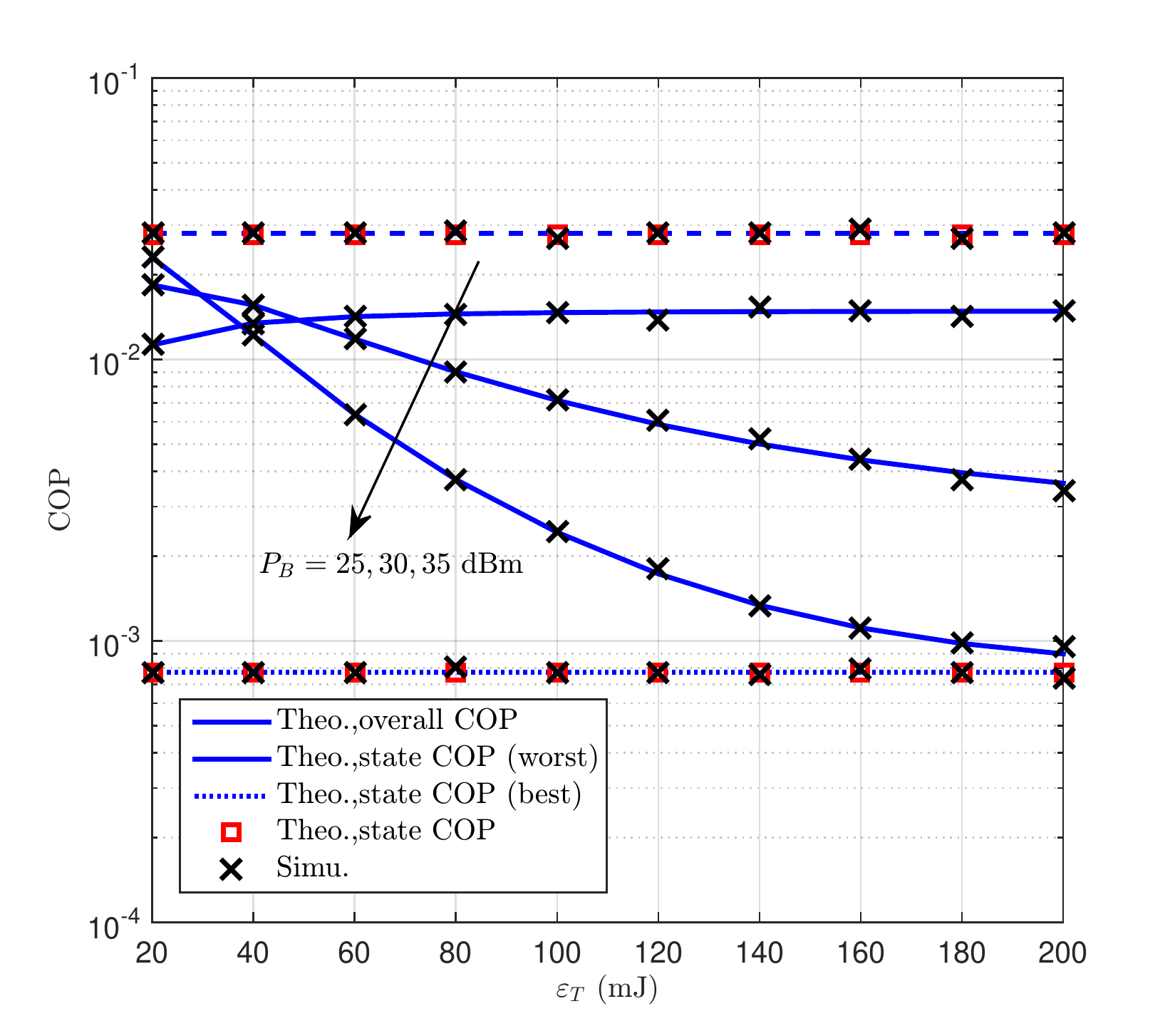}\\
  \caption{COP of the multi-source WPT network versus the energy capacity $\varepsilon_T$ with different $P_B$. $K=2$, $\varepsilon_{\Delta}=20$ mJ, $l_S=l_S^{th}$, $\eta=0.8$, $x_B=-3$ m, $x_D=200$ m, and $r_S=1$ m.}\label{figCOPXenergyT_PBchange_energydelta20_M2}
\end{center}
\end{figure}

\begin{figure}
\begin{center}
  \includegraphics[width=3.5in,angle=0]{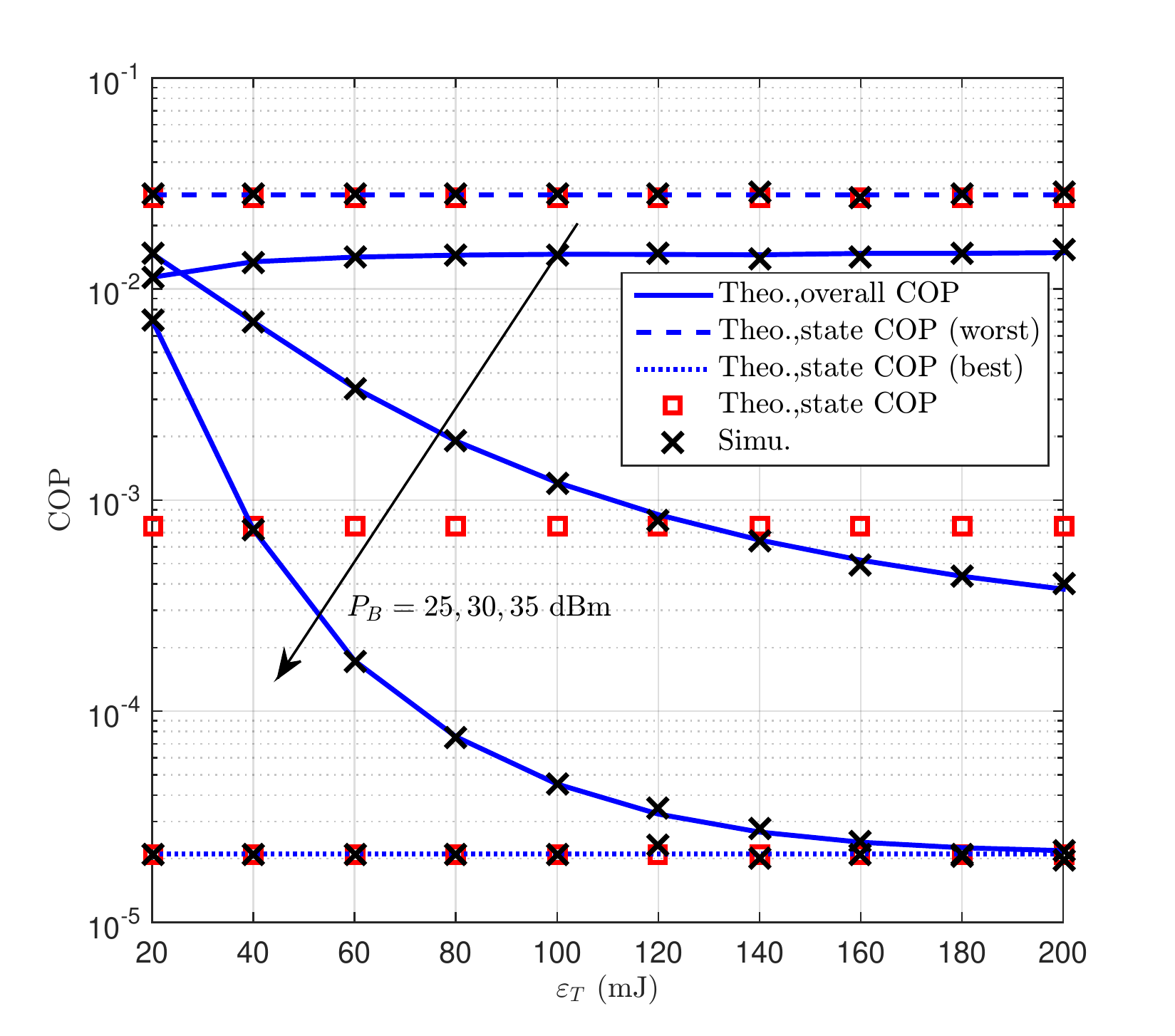}\\
  \caption{COP of the multi-source WPT network versus the energy capacity $\varepsilon_T$ with different $P_B$. $K=3$, $\varepsilon_{\Delta}=20$ mJ, $l_S=l_S^{th}$, $\eta=0.8$, $x_B=-3$ m, $x_D=200$ m, and $r_S=1$ m.}\label{figCOPXenergyT_PBchange_energydelta20_M3}
\end{center}
\end{figure}

Figs. \ref{figCOPXenergyT_PBchange_energydelta20_M2}-\ref{figCOPXenergyT_PBchange_energydelta20_M3} compare the COP of the multi-source WPT network versus the energy capacity $\varepsilon_T$ with different $P_B$. It is noted that the red square symbols represent the COPs when the network remains at a certain state, and the blue lines are the network overall COPs. As depicted in these figures, the COPs under different states vary greatly, and the overall COP firstly approaches to the worst state performance and then goes down to get close to the best state performance if an appropriate $P_B$ could be provided. Besides, this trend could be accelerated by increasing $P_B$. All these observations indicate that a greater energy capacity and $P_B$  are both benefit to decrease the overall COP of the network.

\begin{figure}
\begin{center}
  \includegraphics[width=3.5in,angle=0]{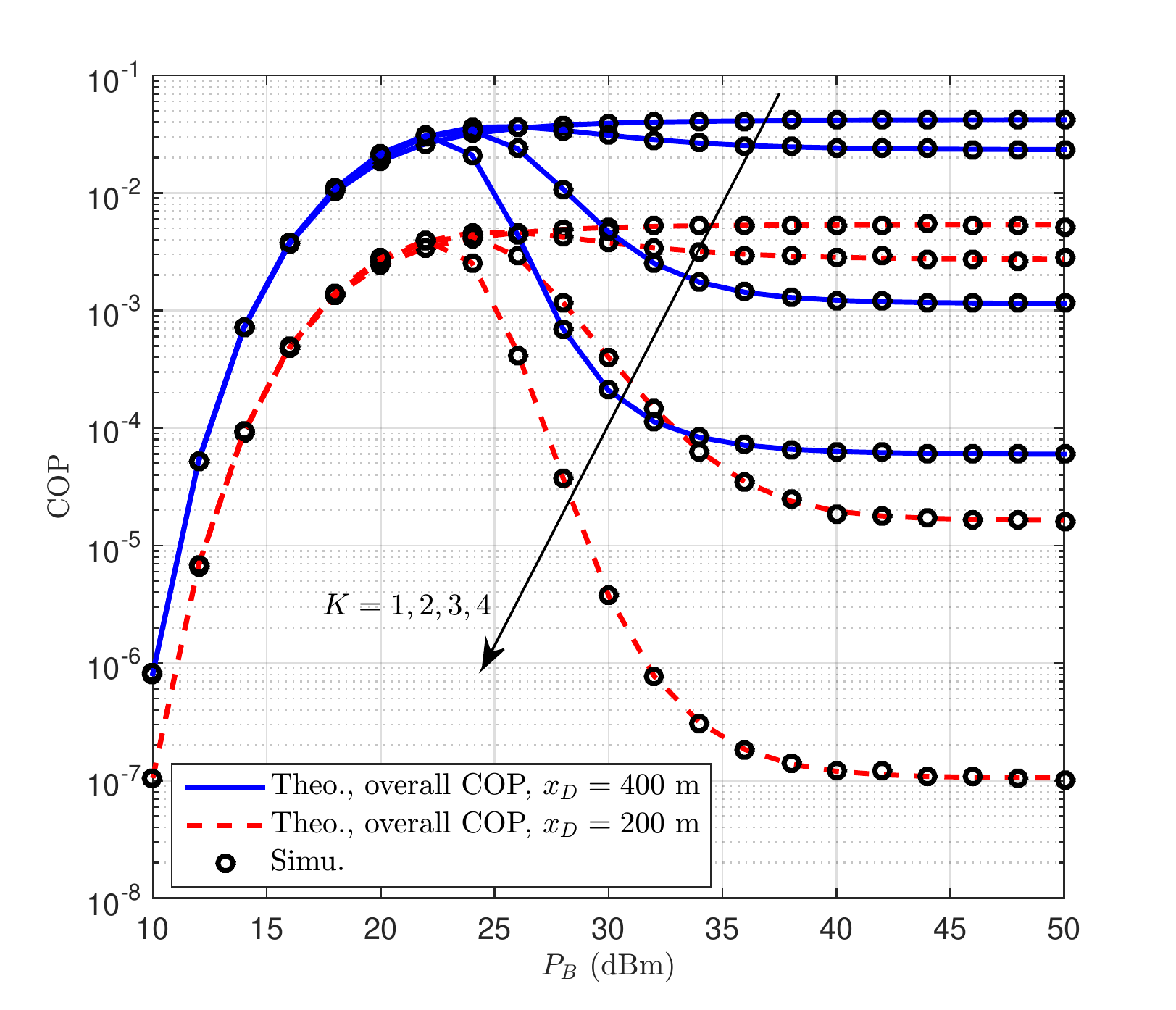}\\
  \caption{COP of the multi-source WPT network versus the transmit power of  power beacon $P_B$ with different $K$ and $x_D$. $L=2$, $\varepsilon_{T}=20$ mJ, $l_S=l_S^{th}$, $\eta=0.8$, $x_B=-3$ m, $x_D=200$ m, and $r_S=1$ m.}\label{figCOPXvecPB_Mchange_dOD_change}
\end{center}
\end{figure}

Fig. \ref{figCOPXvecPB_Mchange_dOD_change} plots the COP of the multi-source WPT network versus the transmit power of the power beacon $P_B$ with different $K$ and $x_D$.  As can be predicted, the COP of the network grows rapidly when $x_D$ increases, which again is resulted from the path loss effect of the wireless channel. In addition, we observe that the COP performance could also be significantly enhanced by adding the number of the sources. It is noted that for a specific line with fixed $K$ and $x_D$, the COP goes up with the increase of $P_B$ at first and then turns down quickly at about 20 to 25 dBm, which reaches a floor eventually. We highlight that, as has clarified previously, it is rather difficult for the sources to collect sufficient energy from the wireless signals if $P_B$ remains at a very low level. Hence, the EOP of network would be rather large in this case, so IT operation can only occur with a very little probability. Recalling that the overall COP is the weighted average of all the states, hence, it would be rather low because the network will stay in energy outage state with a very high probability. It should be pointed out that the low level of COP under this condition does not mean a good performance. Instead, it indicates a very poor performance because it will result in a huge transmission delay to the sources.

\begin{figure}
\begin{center}
  \includegraphics[width=3.5in,angle=0]{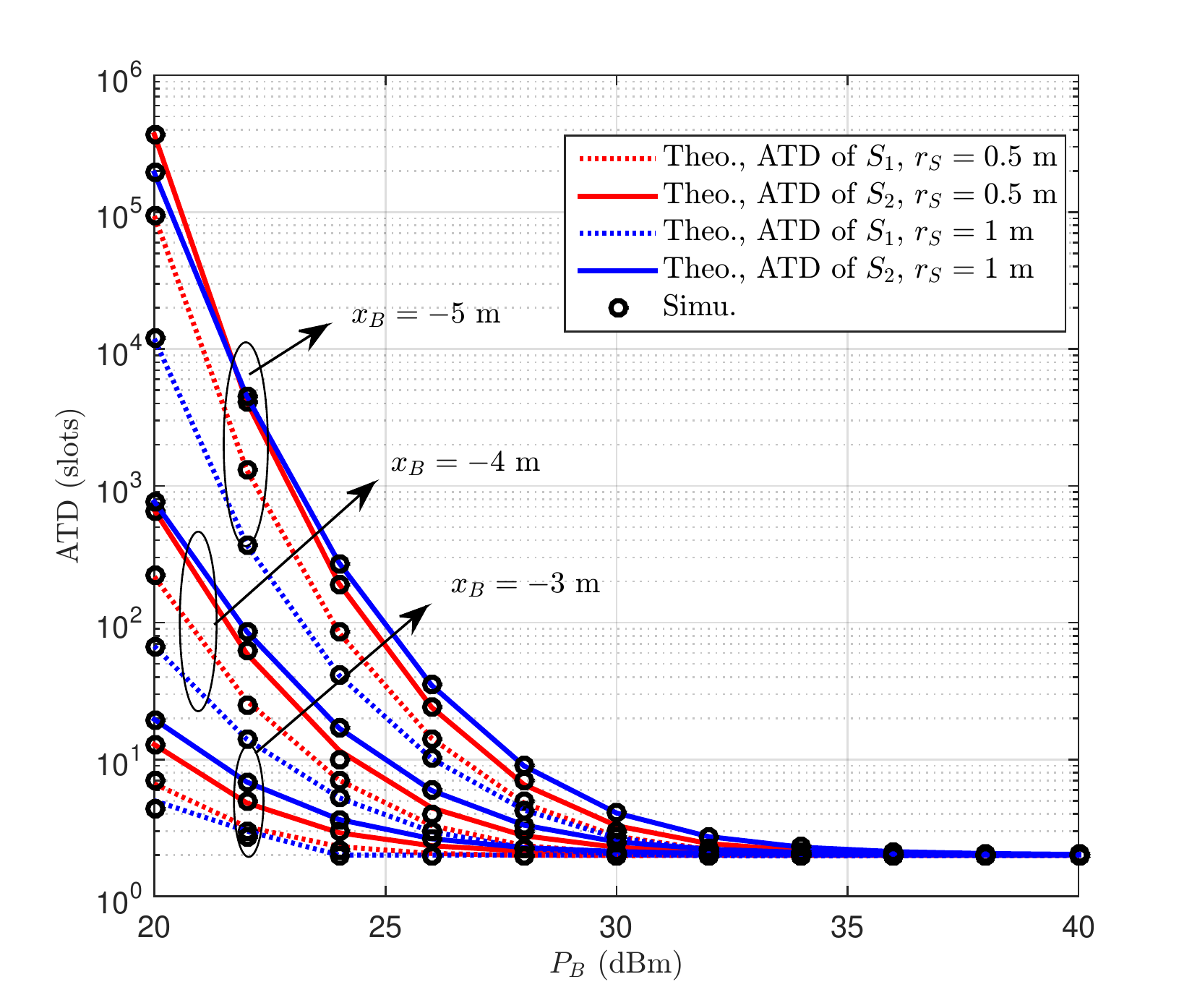}\\
  \caption{ATD of the multi-source WPT network versus the transmit power of  power beacon $P_B$ with different $r_S$ and $x_B$. $K=2$, $L=2$, $\varepsilon_{T}=20$ mJ, $l_S=l_S^{th}$, $\eta=0.8$, and $x_D=200$ m.}\label{figTDXPB_xBchange}
\end{center}
\end{figure}

\begin{figure}
\begin{center}
  \includegraphics[width=3.5in,angle=0]{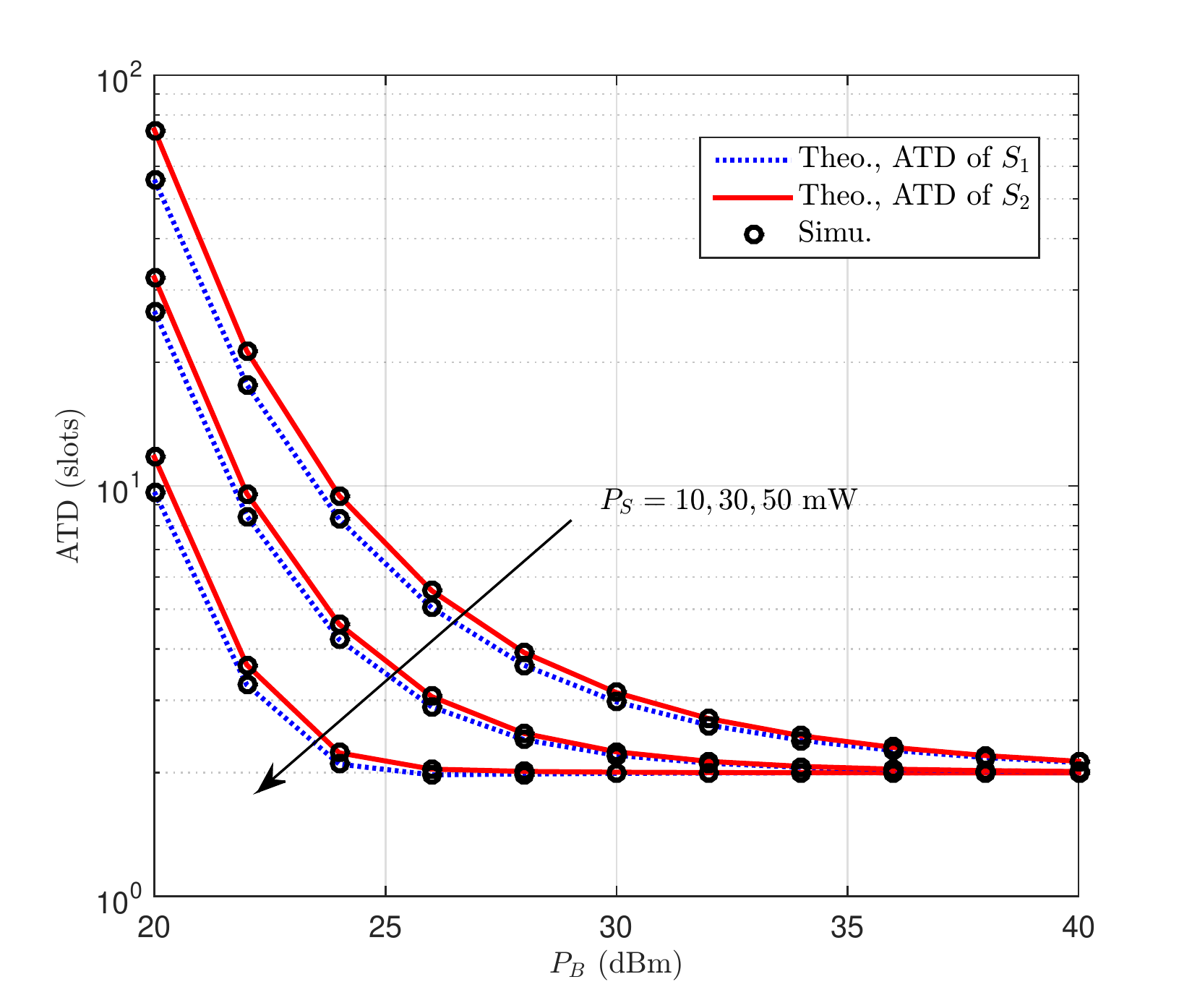}\\
  \caption{ATD of the multi-source WPT network versus the transmit power of  power beacon $P_B$ with different $P_S$. $K=2$, $L=5$, $\varepsilon_{T}=50$ mJ, $\eta=0.8$, $r_S=0.3$ m, $x_B=-3$ m and $x_D=200$ m.}\label{figTDXPB_PSchange}
\end{center}
\end{figure}

Figs. \ref{figTDXPB_xBchange}-\ref{figTDXPB_PSchange} present the ATD of the multi-source WPT network versus the transmit power of the power beacon $P_B$ with different $r_S$, $x_B$ and $P_S$. It is easy to find from these two figures that the ATD performance is not symmetric to all the sources, and this asymmetry  would be enlarged when $r_S$ increases. This is comprehensible because in the proposed network, each source  undergoes independent but not identically distributed channels. Generally speaking, the sources that are more close to the power beacon will have lower ATD. Furthermore, we observe in both two figures that the ATD becomes about 2 time slots when $P_B$ becomes large, which is equal to the number of the sources. Moreover, we see that the ATD rises sharply when the power beacon gets far from the sources. Furthermore, Fig. \ref{figTDXPB_PSchange} depicts that the ATD  for all the sources would increase when the transmit power of sources is promoted.
We note that, under the given conditions, $P_S=10,30,50$ mJ actually corresponds to $l_S=1,3,5$, respectively. As a result, by promoting $P_S$, on the one hand, it needs to spend much more time slots for the sources to accumulate sufficient energy, and on the other hand, the energy consumption of the IT operation would also increase.

\begin{figure}
\begin{center}
  \includegraphics[width=3.5in,angle=0]{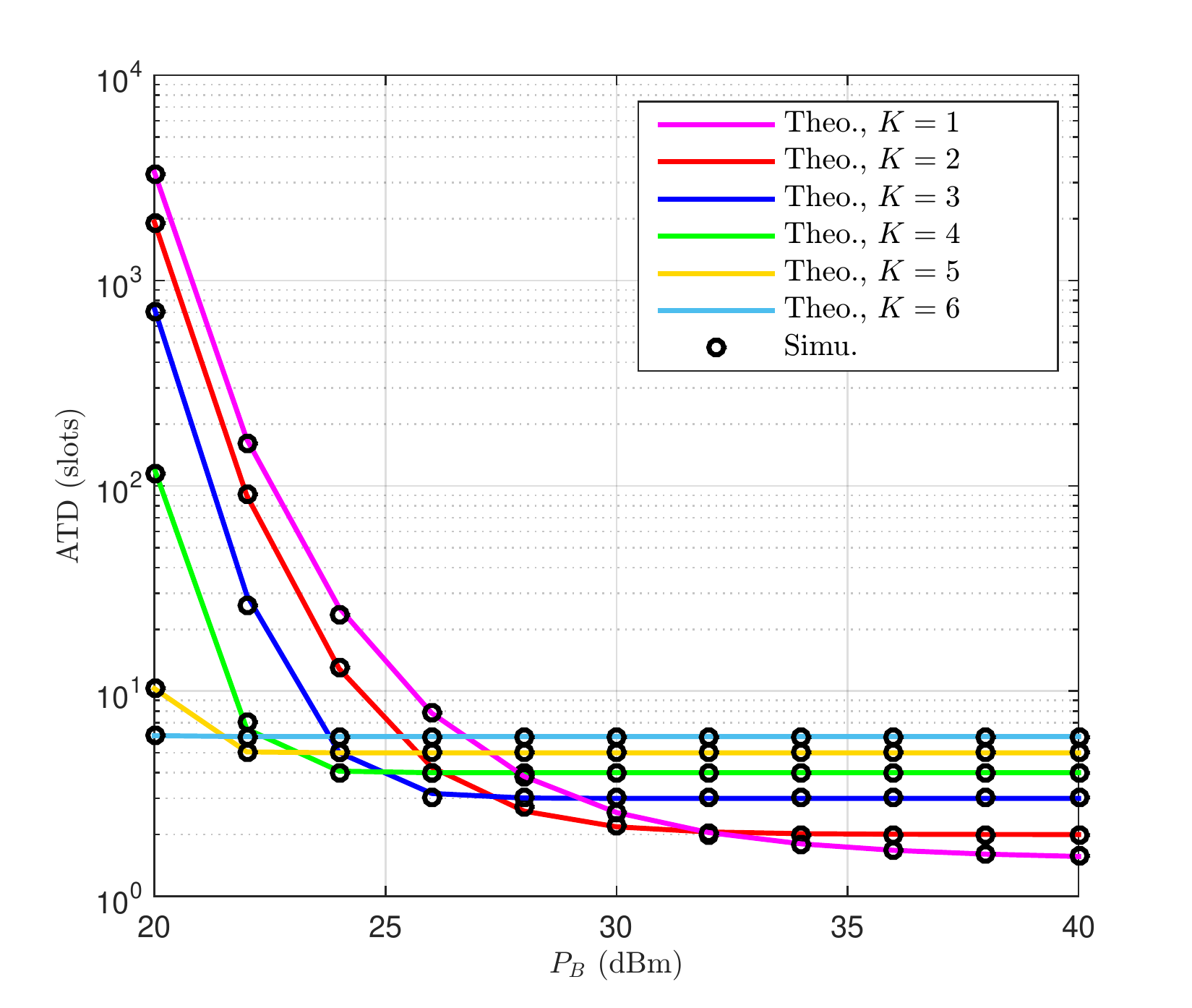}\\
  \caption{ATD of the multi-source WPT network versus the transmit power of  power beacon $P_B$ with different $K$. $L=2$, $\varepsilon_{T}=20$ mJ, $\eta=0.8$, $r_S=0$ m, $x_B=-4$ m and $x_D=200$ m.}\label{figTDXPB_Mchange}
\end{center}
\end{figure}

Figs. \ref{figTDXPB_Mchange} plots the ATD of the multi-source WPT network versus the transmit power of the power beacon $P_B$ with different $K$. Similar with Figs. \ref{figTDXPB_xBchange}-\ref{figTDXPB_PSchange},  we find that the ATD  could be rather huge in the low regime of $P_B$. However, when the number of the sources increase gradually, the ATD performance could be  improved drastically. For example, when $P_B=20$ dBm, the ATD will decline from  about 3000 time slots when $K=1$ to just about 6 time slots when $K=6$. Furthermore, when $P_B$ gets high, the ATD reduces quickly and eventually reaches a constant, which is about $K$. We note that the best ATD performance is also $K$ in a network where all the sources are energy-sufficient, which is resulted from the source selection approach. All above results imply the validity to improve ATD performance by deploying more sources in the network, especially when the wireless energy is not so sufficient.

\section{Conclusions}\label{sec7 Conclusions}

In this paper, we presented a general Markov-based model for the  PB assisted multi-source wireless-powered network with our proposed source selection transmission scheme, which captures the dynamic energy behaviors of the state transitions of the whole network. Two network operating modes, the IT mode with ZF beamforming and the non-IT mode with equal power transmision, were proposed for sustainable energy utilization and reliable data transmission. To characterize the reliability of proposed network, the energy outage probabily was derived for non-IT mode, and the connection outage probability was derived for IT mode. To quantify the delay brought by the source selection transmission, the ATD was also defined and derived. All the analytical results are validated by simualtion, and the results shown that the EOP, COP, and ATD  can be significantly improved via increasing the number of sources deployed in the proposed network.

\appendices



\end{document}